\documentclass[pra,ondcolumn,superscriptaddress]{revtex4}
\usepackage[a4paper, left=1.0in, right=1.0in, top=1.0in, bottom=1.0in]{geometry}

\usepackage[utf8]{inputenc}
\usepackage[english]{babel}
\usepackage[T1]{fontenc}
\usepackage{amsmath}
\usepackage{amsfonts}
\usepackage{hyperref}
\allowdisplaybreaks

\usepackage{tikz}
\usepackage{braket}
\usepackage{natbib}
\usepackage{amsthm}

\newtheorem{theorem}{Theorem}
\newtheorem{lemma}{Lemma}
\newtheorem{fact}{Fact}
\DeclareMathOperator{\Tr}{tr}
\DeclareMathOperator{\E}{\mathbb{E}}

\newcommand{\norm}[1]{\left\lVert#1\right\rVert}

\usepackage[noend]{algpseudocode}
\usepackage{algorithm,algorithmicx}

\algrenewcommand\alglinenumber[1]{\sf\scriptsize\color{blue}{#1}}
\algrenewcommand\algorithmicrequire{\textbf{Input:}}
\algrenewcommand\algorithmicensure{\textbf{Output:}}

\begin{document}

\title{Predicting Features of Quantum Systems from Very Few Measurements}
\date{\today}
\author{Hsin-Yuan Huang}
\email{hsinyuan@caltech.edu}
\affiliation{Institute for Quantum Information and Matter, California Institute of Technology, USA}
\affiliation{Department of Computing and Mathematical Sciences, California Institute of Technology, USA}
\author{Richard Kueng}
\affiliation{Institute for Quantum Information and Matter, California Institute of Technology, USA}
\affiliation{Department of Computing and Mathematical Sciences, California Institute of Technology, USA}

\begin{abstract}
Predicting features of complex, large-scale quantum systems is essential to the characterization and engineering of quantum architectures. 
We present an efficient approach for constructing an approximate classical description, called the \emph{classical shadow}, of a quantum system from very few quantum measurements that can later be used to predict a large collection of features. This approach is guaranteed to accurately predict $M$ linear functions with bounded Hilbert-Schmidt norm from only order of $\log (M)$ measurements. This is completely independent of the system size and saturates fundamental lower bounds from information theory.
We support our theoretical findings with numerical experiments over a wide range of problem sizes (2 to 162 qubits). These highlight advantages compared to existing machine learning approaches. 
\end{abstract}

\maketitle

\section{Introduction and Main Results}

\subsection{Motivation}

Making predictions based on empirical observations is a central subject studied in statistical learning theory and is at the heart of many scientific disciplines, including quantum mechanics.
There, predictive tasks, like estimating target fidelities, or verifying entanglement, are essential ingredients for manufacturing, calibrating and controlling quantum architectures. 
Recently, unprecedented advances in the size of such architectures \cite{preskill2018nisq} has pushed traditional prediction techniques -- like quantum state tomography -- to the limit of their capabilities. This is mainly due to a curse of dimensionality: the number of parameters that describe a quantum system scales exponentially in the number of constituents. Moreover, these parameters cannot be accessed directly, but must be estimated by measuring the system. An informative quantum mechanical measurement is both destructive (wave-function collapse) and only yields probabilistic outcomes (Born's rule). Hence, many samples are required to accurately estimate even a single parameter of the underlying quantum system. Furthermore, all of these outcomes must be processed and stored in memory for subsequent prediction of relevant features.
In summary: reconstructing a full description of a quantum system with $n$ constituents (e.g. qubits) necessitates an exponential number of measurement repetitions, and an exponential amount of classical memory and computing power. 

Several approaches have been proposed to overcome this fundamental scaling problem. These include matrix product state (MPS) tomography \cite{cramer2010efficient} and neural network tomography \cite{torlai2018neural, carrasquilla2019reconstructing}. Both only require a polynomial number of samples, provided that the underlying state have advantageous structure. However, for general quantum systems, these techniques still require an exponential number of samples.
We refer to the related work section for details.

Pioneering a conceptually very different line of research, Aaronson et al. \cite{aaronson2018shadow, aaronson2019gentle} pointed out that demanding full classical descriptions of quantum systems may be excessive for many concrete tasks. Instead it is often sufficient to accurately predict certain features. In quantum mechanics, interesting features are typically linear functions in the underlying density matrix $\rho$:
\begin{equation}
o_i (\rho) = \mathrm{trace}(O_i \rho) \quad 1 \leq i \leq M. \label{eq:linear_predictions}
\end{equation}
The fidelity with a pure target state, entanglement witnesses, potential future measurement statistics and expectation values of observables are but a few prominent examples.
Aaronson coined the term
\cite{aaronson2018shadow} \emph{shadow tomography}\footnote{According to Ref.~\cite{aaronson2018shadow} it was actually S.T.\~Flammia who originally suggested the name shadow tomography.} for this concrete task and showed that already a polynomial number of state copies suffice to predict an exponential number of target functions. While very efficient in terms of samples, this procedure is very demanding in terms of quantum hardware: a concrete implementation of the existing protocol requires exponentially long quantum circuits that act simultaneously on all copies of the unknown state stored in a quantum memory.

\begin{figure*}[t]
    \centering
    \includegraphics[width=1.0\textwidth]{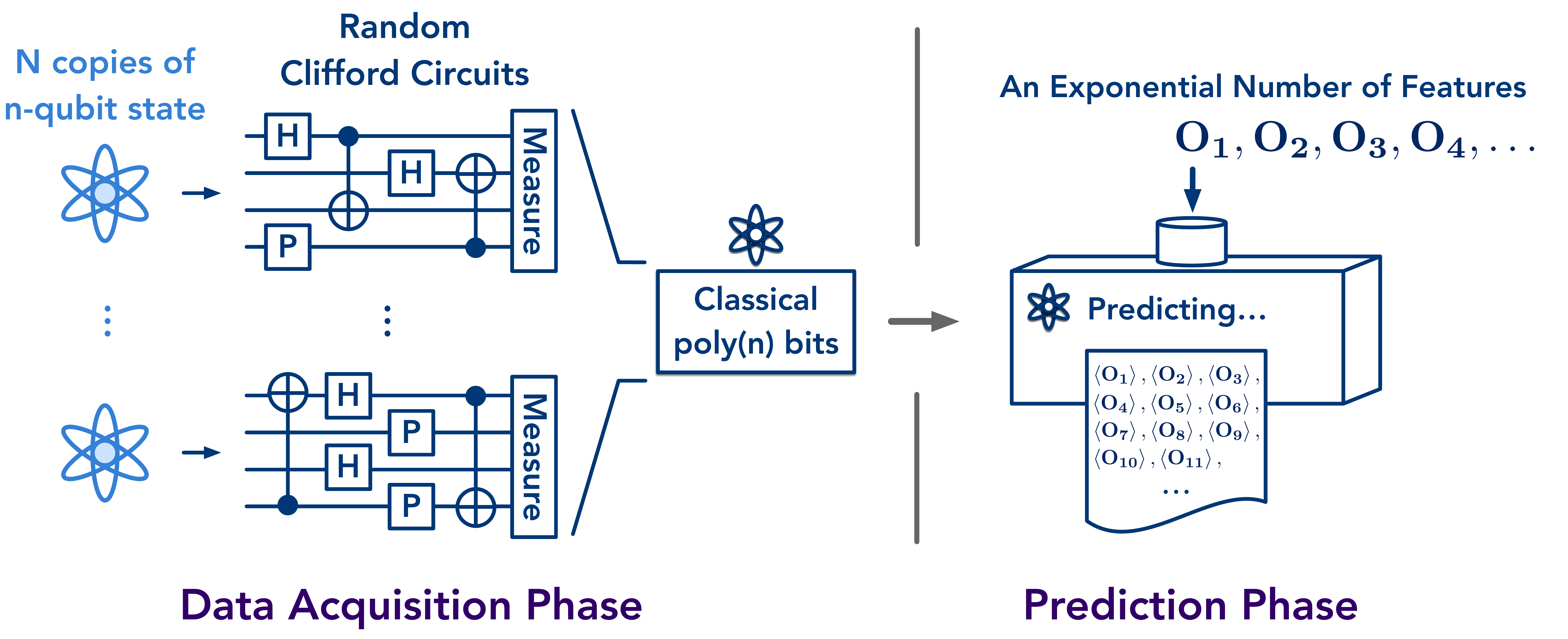}
    \caption{\emph{Caricature of classical shadows:} In the data acquisition phase, we perform random Clifford measurements on independent copies of a $n$-qubit system to obtain a classical representation of the quantum system -- the \emph{classical shadow}. Such classical shadows facilitate accurate prediction of a large number of features using a simple median-of-means protocol. }
    \label{fig:setup}
\end{figure*}

In this work, we combine the mindset of shadow tomography \cite{aaronson2018shadow} (predict linear target functions, not the full state)  with recent insights from quantum state tomography \cite{guta2018fast} (rigorous statistical convergence guarantees)  and the  stabilizer formalism \cite{gottesman1997stabilizer} (efficient implementation).
The result is a highly efficient protocol that first learns a minimal classical sketch $S_\rho$ -- the \emph{classical shadow} -- of an unknown quantum state $\rho$. This sketch can be stored efficiently and is based on independent, tractable quantum measurements. 
Subsequently $S_\rho$ can be used to predict arbitrary linear function values \eqref{eq:linear_predictions} by a simple median-of-means protocol. We refer to Figure~\ref{fig:setup} for an illustration.

Our main technical contribution equips this procedure with rigorous performance guarantees. A classical shadow of polynomial size provably suffices to accurately predict an exponential number of target functions. More precisely, we only need a classical shadow of size $B \log(M) / \epsilon^2$ to predict $M$ target functions, provided that $B \geq \max_i \Tr(O_i^2)$. We emphasize that this rigorous result establishes an exponential compression both in terms of dimensionality and the number of target functions.
The scaling with respect to~$B$ can be restrictive in certain cases.
However, we show that this is not a flaw of the method, but an unavoidable limitation rooted in information theory.
By relating feature prediction to a communication task whose potential is limited by information theory~\cite{fano1961information_theory},
we establish a fundamental lower bound of $B \log(M) / \epsilon^2$, which matches with the performance of classical shadows.

In contrast to Aaronson's prior results, the quantum and classical aspects of this procedure are manifestly tractable -- both in theory and practice.
We support this claim by conducting numerical simulations for quantum systems comprised of up to 162 qubits. 
We observe that prediction via classical shadows scales favourably and improves upon existing techniques (neural network tomography)
in a variety of practically-motivated test cases.

\subsection{The procedure}

\paragraph{Data acquisition phase:}

Throughout this work we restrict attention to multi-qubit systems and $\rho$ is a fixed, but unknown, quantum state in $D=2^n$ dimensions.
To extract meaningful information, we repeatedly perform a simple measurement procedure: randomly rotate the state ($\rho \mapsto U \rho U^\dagger$) and perform a computational basis measurement. 
Upon receiving outcome $\hat{b} \in \left\{0,1 \right\}^n$, we store a classical description of $U^\dagger |\hat{b} \rangle $ in the classical memory. Repeating this procedure $N$ times results in an array of $N$ rotated basis states:
\begin{equation}
\mathsf{S}(\rho;N) = \left\{ U_1^\dagger | \hat{b}_1 \rangle,\cdots U_N^\dagger | \hat{b}_N \rangle \right\}. \label{eq:classical_shadow}
\end{equation}
We call this array the \emph{classical shadow} of $\rho$.
To minimize storage cost, we choose each $U_i$ uniformly from the Clifford group \cite{gottesman1997stabilizer}. This group is generated by CNOT, Hadamard and Phase gates and has an extensively rich and well-understood structure. In particular, each $U_i | \hat{b}_i \rangle$ is a stabilizer state that is fully characterized by only $\mathcal{O}(n^2)$ bits \cite{aaronson2004improved}. What is more, Cifford measurements are known to be practically feasible in a large variety of quantum architectures.

\paragraph{Prediction phase:}

We take guidance from quantum state tomography \cite{sugiyama2013tomography,guta2018fast}. 
The classical shadow \eqref{eq:classical_shadow} may be interpreted as outcome/frequency statistics associated with a single quantum measurement (the POVM of all stabilizer states). The associated linear inversion estimator is $\hat{\rho} = (2^n+1) / N \sum_{i=1}^N U_i^\dagger |\hat{b}_i \rangle \! \langle \hat{b}_i| U_i - \mathbb{I}$.
Rather than forming a single estimator, we split the classical shadow up into  equally-sized chunks and construct several, independent linear inversion estimators.
Subsequently, we predict linear function values \eqref{eq:linear_predictions} 
via \emph{median of means estimation} \cite{jerrum1986medianmeans,nemirovski1983medianmeans}, see also \cite{gosset2018compressed}. This procedure is summarized in Algorithm~\ref{alg:median_means}.

\begin{algorithm*}[t]
{\small
\begin{algorithmic}[1]
\caption{{\small \textit{Median of means prediction}  based on a classical shadow $\mathsf{S}(\rho,N)$.}}
\label{alg:median_means}

\Statex
\Function{LinearPredictions}{$O_1,\ldots,O_M,\mathsf{S}(\rho;N),K$}

\State Import $\mathsf{S}(\rho;N) = \left[ U_1^\dagger | \hat{b}_1 \rangle,\ldots, U_N^\dagger | \hat{b}_N \rangle \right]$
\Comment Load classical shadow
\State Split the shadow into $K$ equally-sized parts and set
\Comment Construct $K$ linear estimators of $\rho$
\begin{displaymath}
\hat{\rho}_{(k)} = \frac{2^n+1}{\lfloor N/K \rfloor} \sum_{i = (k-1) \lfloor N/K \rfloor+1}^{k\lfloor N/K \rfloor} U_i^\dagger | \hat{b}_i \rangle \! \langle \hat{b}_i | U_i - \mathbb{I} \quad \textrm{for} \quad 1 \leq k \leq K.
\end{displaymath}

\For{$i=1$ to $M$}

\State Output
$
\hat{o}_i (N,K) = \mathrm{median} \left\{ \mathrm{tr} \left(O_i \hat{\rho}_{(1)} \right),\ldots, \mathrm{tr} \left( O_i \hat{\rho}_{(K)} \right) \right\}.
$
\Comment Median of means estimation
\EndFor
\EndFunction
\end{algorithmic}
}
\end{algorithm*}

\subsection{Rigorous performance guarantees}
\label{sec:main}


\begin{theorem}[informal version] \label{thm:main}
Classical shadows of size (order) $\log (M) \max_i \mathrm{tr}(O_i^2)/\epsilon^2$  suffice to predict $M$ linear target functions $\mathrm{tr}(O_i \rho)$ up to accuracy $\epsilon$.
\end{theorem}

We refer to Appendix~\ref{sec:thm1proof} for a detailed statement and proofs. 
There we also report substantial improvements in the case of very many predictions, i.e.\ $M >2^{2^n}$.
Theorem~\ref{thm:main} is most powerful, when the linear functions have constant Hilbert-Schmidt norm. In this case, classical shadows allow for predicting an exponential number of linear functions based on a polynomial number of quantum measurements only. Moreover, this sampling rate is completely independent of the system dimension.

\subsection{Illustrative example applications}

\paragraph{Direct fidelity estimation:}

Suppose that we wish to certify that an experimental device prepares a desired $n$-qubit state. Typically, this target state $| \psi \rangle \! \langle \psi|$ is pure and highly structured, e.g.\ a a GHZ state \cite{greenberger1989GHZ} for quantum communication protocols, or a 
toric code ground state \cite{kitaev2002toric} for fault-tolerant quantum computation.
Theorem~\ref{thm:main} asserts that a classical shadow of dimension-independent size suffices to accurately predict the fidelity of \emph{any} state in the lab with \emph{any} pure target state. This base case improves over the best existing result on direct fidelity estimation \cite{flammia2011direct} which requires $O(2^n / \epsilon^4)$ samples in the worst case. Moreover, a classical shadow of polynomial size allows for estimating an exponential number of (pure) target fidelities all at once.

\paragraph{Entanglement verification:} 

Fidelities with pure target states can also serve as (bipartite) \emph{entanglement witnesses} \cite{guehne2009entanglement}.
Separable $n$-qubit states $\rho$ necessarily obey $\langle \psi| \rho | \psi \rangle \leq 2^{-n / 2}$ for any maximally entangled state $| \psi \rangle \! \langle \psi|$.
In turn, entanglement verification may be achieved by finding a single maximally entangled state with $\langle \psi| \rho | \psi \rangle >2^{-n /2}$. Classical shadows of logarithmic size allow for checking an exponential number of such entanglement witnesses simultaneously. What is more, the actual (Clifford) measurement procedure is completely independent of the witnesses in question. There is no need to locally decompose \cite{guehne2009entanglement} a concrete witness in order to accurately estimate it directly. These ideas readily extend to genuine multipartite entanglement detection.

\paragraph{Predicting expectation values of certain observables (non-example):}

The applicability of Theorem~\ref{thm:main} is not without limitations. The size of classical shadows must scale with $\mathrm{tr}(O_i^2)$ -- the (maximal) Hilbert-Schmidt norm among target functions. This scaling is benign for observables that are low rank, but can become exponentially large for certain observables.
A concrete example is Pauli expectation value of a spin chain: $\langle \sigma_{i_1} \otimes \cdots \otimes \sigma_{i_n} \rangle_\rho = \mathrm{tr} \left( O_1 \rho \right)$, and $\mathrm{tr}(O_1^2) = \mathrm{tr} \left( \mathbb{I}^{\otimes n} \right) =2^n$. In this case, a classical shadow of exponential size 
is required to accurately predict a single expectation value. In contrast, a direct spin measurement achieves the same accuracy with an order of $1/\epsilon^2$ copies of the state $\rho$ only. However, even in this regime, classical shadows still provide some advantages.
A shadow of size $n  2^n /\epsilon^2$ allows for predicting all possible spin expectation values simultaneously -- a square-root improvement over naively measuring all $4^n$ different spin combinations directly.
This scaling coincides with an optimal measurement strategy that groups all Pauli strings into $2^n+1$ groups each containing $2^n-1$ fully commuting Pauli strings \cite{yen2019measuring, crawford2019efficient, jena2019pauli, huggins2019efficient}. Furthermore, we will later show that we can also predict $M$ few-body Pauli observables using $\log(M)$ measurements.

\subsection{Matching information-theoretic lower bounds}

The non-example from above raises an important question: is the scaling with $\mathrm{tr}(O_i^2)$ fundamental, or a mere artifact of feature prediction with classical shadows?
A rigorous analysis reveals that this scaling is no mere artifact, but reflects fundamental restrictions from information theory.

\begin{theorem}[informal version] \label{thm:main2}
\emph{Any} linear prediction procedure based on a fixed set of independent measurements requires at least $\log (M) \max_i \mathrm{tr}(O_i^2)/\epsilon^2$ state copies.
\end{theorem}


This fundamental result highlights that feature prediction with classical shadows is asymptotically optimal.
We refer to Appendix~\ref{sec:proofthm2} for details and proofs. 



There are a few ways to bypass this no-go theorem. One is performing joint quantum measurements on all copies of the quantum state $\rho$ simultaneously. Another is to use measurements that are attuned to the particular target functions in question.
For example, Aaronson's shadow tomography protocol avoids the scaling with $ \max_i \Tr(O_i^2)$ by utilizing exponential-size quantum circuits that depend on the $M$ target functions and act on all copies of the unknown state simultaneously.

\section{Numerical experiments}

\subsection{Direct fidelity estimation and comparison with machine learning approaches}

\begin{figure}[t]
    \centering
    \includegraphics[width=1.0\textwidth]{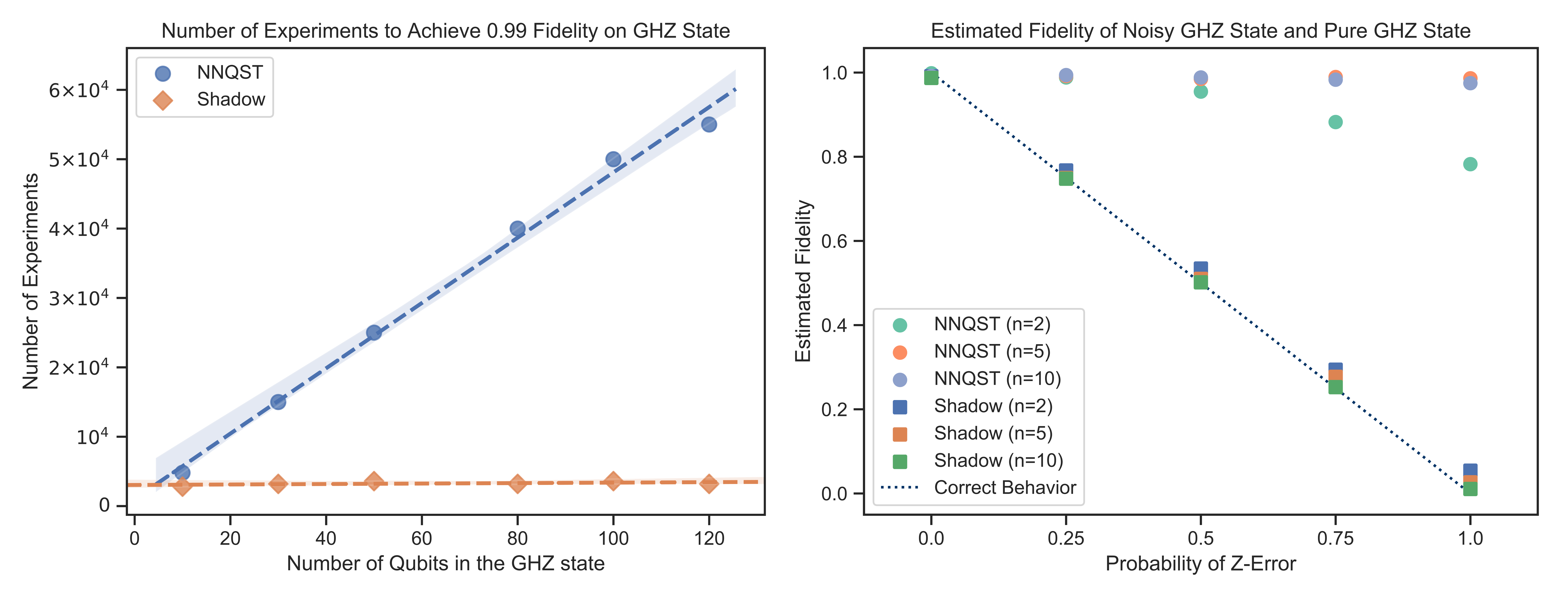}
    \caption{\emph{Comparison between classical shadow and neural network tomography (NNQST); GHZ states.} \\
\emph{Left:} Number of measurements required to identify an $n$-qubit GHZ state with 0.99 fidelity. The shaded regions are the standard deviation of the needed number of experiments over ten independent runs. \\
\emph{Right:} Estimated fidelity between a perfect GHZ target state and a noisy preparation, where $Z$-errors can occur with probability $p \in [0,1]$, under $6\times 10^4$ experiments. The dotted line represents the true fidelity as a function of $p$. \\
NNQST can only estimate classical fidelity (which is an upperbound of quantum fidelity) efficiently, so we consider classical fidelity for NNQST but remain using quantum fidelity for classical shadow.
}
    \label{fig:tomoGHZ}
\end{figure}

One of the key features of classical shadows is scalability. For a large class of features -- such as those with efficient stabilizer decomposition -- scalability extends to the computational cost associated with median of means prediction\footnote{The runtime of Algorithm~\ref{alg:median_means} is dominated by the cost of computing squared inner products $| \langle \psi | U_i | \hat{b}_i \rangle |^2$ in $2^n$ dimensions. If $| \psi \rangle$ is a stabilizer state, the Gottsman-Knill theorem allows for evaluation in $\mathcal{O}(n^2)$-time only.}.
This structure allows us to conduct numerical experiments for systems comprised of more than 160 qubits. 
Machine learning based approaches \cite{carrasquilla2019reconstructing, torlai2018neural} are among the most promising methods that still cover such regimes, where the number of degrees of freedom is roughly comparable to the total number of silicon atoms on earth ($2^{162} \simeq 10^{48.8}$).
The most recent version of \emph{neural network quantum state tomography} (NNQST) is a generative model that is based on a deep neural network trained on independent quantum measurement outcomes (local SIC/tetrahedral POVMs \cite{renes2004SIC}). 
We consider the task of learning classical representation of an unknown quantum system, and using the representation to predict fidelity with some target state.

\paragraph{GHZ states}

In \cite{carrasquilla2019reconstructing}, the viability of NNQST is demonstrated by considering GHZ states on a varying number of qubits $n$. Numerical experiments highlight that the number of measurement repetitions (size of the training data) to learn a neural network model of the GHZ state that achieves target fidelity of $0.99$ scales linearly in $n$.
We have also implemented NNQST for GHZ states and compared it to classical shadows with median of means prediction. The left hand side of Figure~\ref{fig:tomoGHZ} confirms the linear scaling of NNQST and the assertion of Theorem~\ref{thm:main}: classical shadows of \emph{constant} size suffice to accurately estimate GHZ target fidelities, regardless of the actual system size. Subsequently, we have also tested the capability of both approaches to detect potential state preparation errors.
More precisely, we consider a scenario where the GHZ-source introduces a phase error with probability $p \in [0,1]$: 
\begin{equation*}
\rho_p = (1-p) |\psi_\mathrm{GHZ}^+(n)\rangle \! \langle \psi_\mathrm{GHZ}^+ (n)|+ p|\psi_\mathrm{GHZ}^-(n)\rangle \! \langle \psi_\mathrm{GHZ}^- (n)|, \quad | \psi_\mathrm{GHZ}^\pm (n) \rangle = \tfrac{1}{\sqrt{2}} \left( |0 \rangle^{\otimes n} \pm |1 \rangle^{\otimes n} \rangle \right).
\end{equation*}
We learn a classical representation of the GHZ-source and subsequently predict the fidelity with the pure GHZ state.
The right hand side of Figure~\ref{fig:tomoGHZ} highlights that classical shadow prediction accurately tracks the decrease in target fidelity as the error parameter increases. NNQST, in contrast, seems to consistently overestimate this target fidelity. 
In the extreme case ($p=1$), the true underlying state is completely orthogonal to the target state, but NNQST nonetheless reports fidelities close to one. 
We believe that this shortcoming arises because the POVM-based machine learning approach can only estimate classical fidelity efficiently, which is an upper bound of the actual quantum fidelity. Also, the existing approach uses local measurements that may be ineffective at capturing global phases.
The same issue holds for the original RBM-based neural network tomography \cite{torlai2018neural} due to the necessity of using local measurements in achieving efficient training.
Potential solutions include devising an effective estimation of quantum fidelity in deep generative models and designing neural networks that can learn from highly non-local measurements.
Achieving these is non-trivial and remains open.

\paragraph{Toric code ground states}

\begin{figure}[t]
    \centering
    \includegraphics[width=1.0\textwidth]{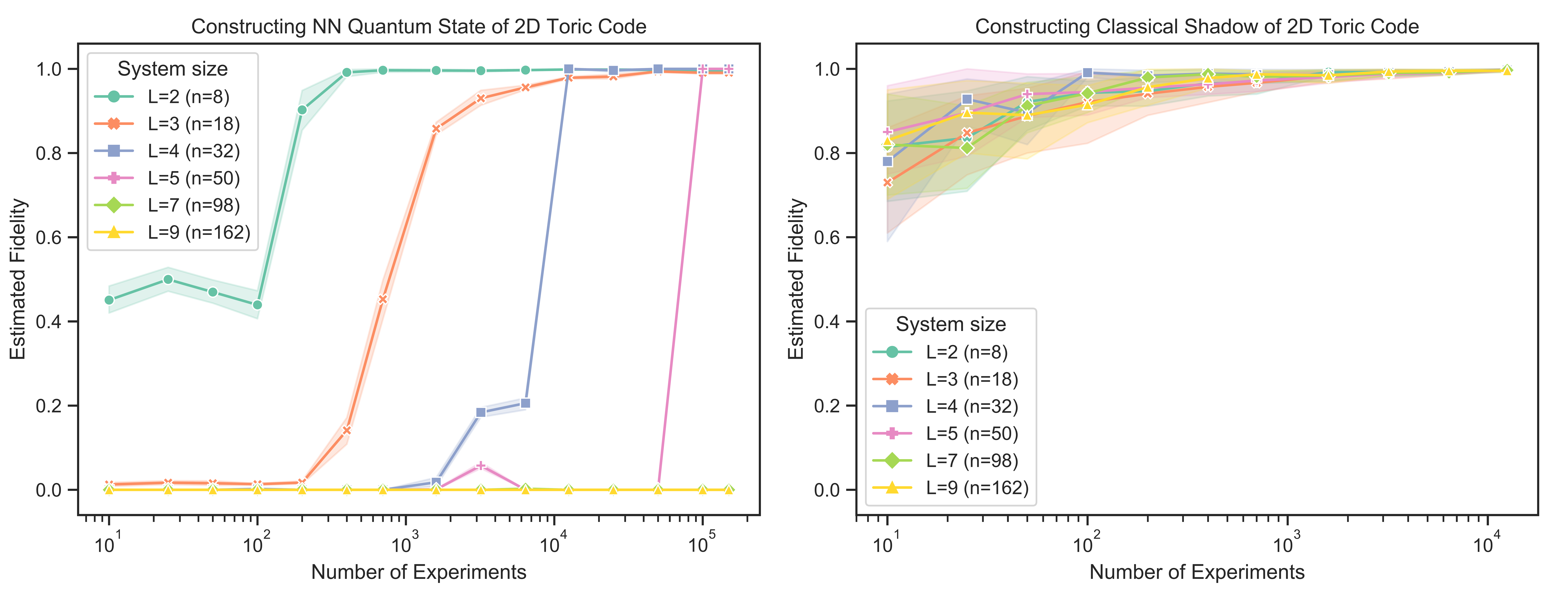}
    \caption{\emph{Comparison between classical shadow and neural network tomography (NNQST); toric code:} \\
\emph{Left:} Number of measurements required for neural network tomography to identify a particular toric-code ground state. We use classical fidelity for NNQST, which is an upper bound for quantum fidelity.
\emph{Right:} Performance of classical shadows for the same problem. We use quantum fidelity for classical shadows.
The shaded regions are the standard deviation of the estimated fidelity over ten runs.
}
    \label{fig:tomoToric}
\end{figure}

While highly instructive from a theoretical perspective, GHZ states comprised of 100 qubits are very fragile and challenging to implement in practice. 
We have also conducted experiments for target states that are more physical. 
\emph{Toric code ground states} serve as an excellent test case. Not only are they
the most prominent example of a topological quantum error correcting code and thus highly relevant for quantum computing devices. They also correspond to ground states of a Hamiltonian: $H=-\sum_v A_v - \sum_p B_p$, where $A_v$ and $B_p$ denote vertex- and plaquette operators\footnote{$A_v$ is the product of four Pauli-$X$ operators around a vertex $v$, while $B_p$ is the product of four Pauli-$Z$ operators around the plaquette $p$.}.
The ground space of $H$ is four-fold degenerate and we select the superposition of all closed-loop configurations ($| \psi \rangle \propto \sum_{S:\textrm{ closed loop}} |S \rangle$) as a test state for both classical shadow and NNQST: how many measurement repetitions are required to accurately identify this toric code ground state with high fidelity? The results are shown in Figure~\ref{fig:tomoToric}. 
Neural network tomography based on deep generative model seems to require a number of samples that scales unfavorably in the system size $n$ (left). In contrast, fidelity estimation with classical shadows is completely independent of the system size.
The difficulty of NNQST in learning 2D toric code may be related to some observed failures of deep learning \cite{shalev2017failures} for learning patterns with combinatorial structures.
In Appendix~\ref{sec:detail}, we provide further analysis of potential difficulty in using machine learning approaches to reconstruct some simple quantum states due to a well-known computational hardness conjecture.

\subsection{Witnesses for tri-partite entanglement}

\begin{figure}[t]
    \centering
    \includegraphics[width=1.0\textwidth]{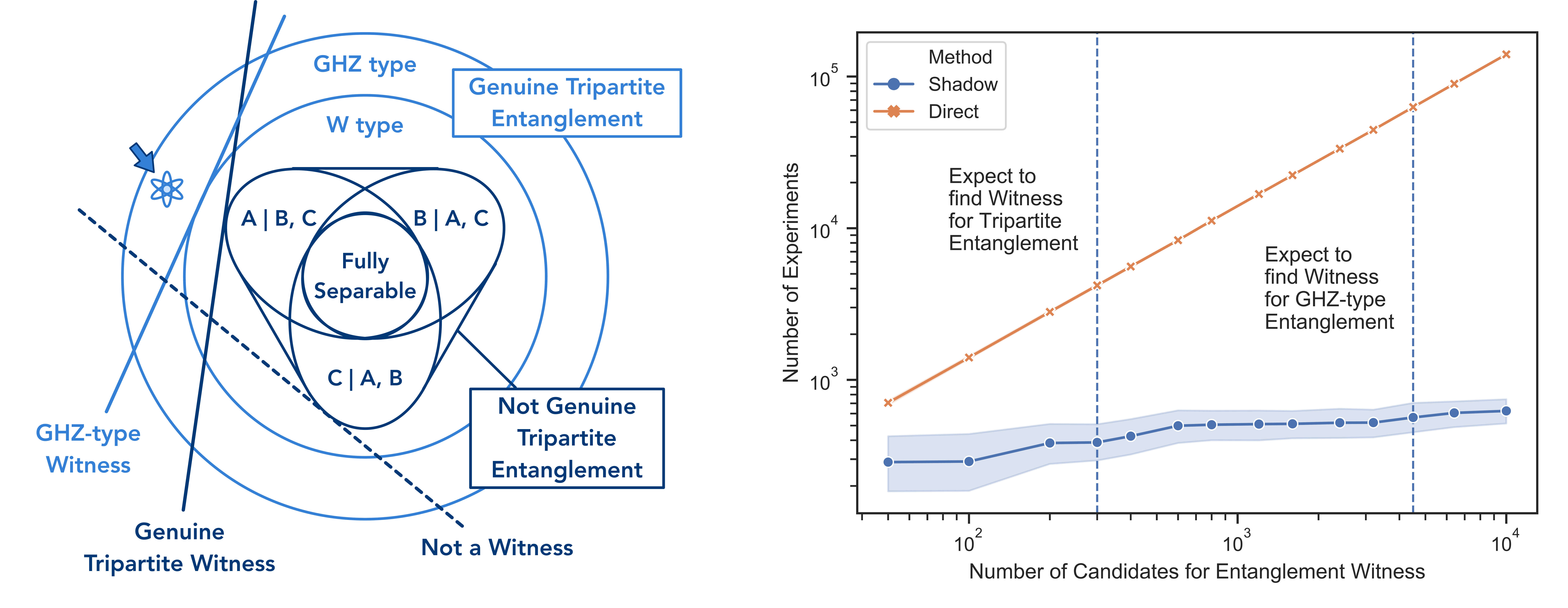}
    \caption{\emph{Detection of GHZ-type entanglement for 3-qubit states:} \\
\emph{Left:} Schematic illustration of 3-partite entanglement. Entanglement witnesses are linear functions that separate part of one entanglement class from all other classes. \\ 
\emph{Right:} Number of entanglement witnesses vs.\ number of experiments required to accurately estimate all of them. The dashed lines represent the expected number of (random) entanglement witnesses required to detect genuine three-partite entanglement and GHZ-type entanglement in a randomly rotated GHZ state. The shaded region is the standard deviation of the required number of experiments over ten independent repetitions of the entire setup.}
    \label{fig:witness}
\end{figure}

Entanglement is at the heart of virtually all quantum communication and cryptography protocols and an important resource for quantum technologies in general. 
This renders the task of detecting entanglement important both in theory and practice \cite{friis2019entanglement,guehne2009entanglement}. While bi-partite entanglement is comparatively well-understood, multi-partite entanglement has a much more involved structure. 
Already for $n=3$ qubits, there is a variety of inequivalent entanglement classes. These include fully-separable, as well as bi-separable states, $W$-type states and finally GHZ-type states. The relations between these classes are summarized in Figure~\ref{fig:witness} and we refer to \cite{acin2001classification} for a complete characterization. Despite this increased complexity, entanglement witnesses remain a simple and useful tool for testing which class a certain state $\rho$ belongs to. 
However, any given entanglement witness only provides a one-sided test -- see Figure~\ref{fig:witness} (left) for an illustration -- and it is often necessary to compute multiple witnesses for a definitive answer. 

Classical shadows can considerably speed up this search: according to Theorem~\ref{thm:main} a classical shadow of moderate size allows for checking an entire list of fixed entanglement witnesses simultaneously. Figure~\ref{fig:witness} (right) underscores the economic advantage of such an approach over measuring the individual witnesses directly. 
Directly measuring $M$ different entanglement witnesses requires a number of  quantum measurements that scales (at least) linearly in $M$. In contrast, classical shadows get by with $\log (M)$-many measurements only.

More concretely, suppose that the state to be tested is a local, random unitary transformation of the GHZ state. Then, this state is genuinely tripartitely entangled and moreover belongs to the GHZ class. 
The dashed vertical lines in Figure~\ref{fig:witness} (right) denote the expected number of (randomly selected) witnesses required to detect genuine tri-partite entanglement (first) and GHZ-type entanglement (later). 
From the experiment, we can see that classical shadows achieve these thresholds with an exponentially smaller number of samples than the naive direct method.
Finally, classical shadows are based on random Clifford measurements and do not depend on the structure of the concrete witness in question. This is another practical advantage over direct estimation, which crucially depends on the concrete witness in question that may be difficult to measure. 

\section{Discussion}

\subsection{Related work}

\paragraph{General quantum state tomography} 

The task of reconstructing a full classical description -- the density matrix $\rho$ -- of a $D$-dimensional quantum system from experimental data is one of the most fundamental problems in quantum statistics, see e.g.\ \cite{hradil1997tomography,blumekohout2010tomography,gross2010compressed,gross2013focus} and references therein. Sample-optimal protocols, i.e.\ estimation techniques that get by with a minimal number of measurement repetitions, have only been developed recently. 
Information-theoretic bounds assert that an order of $\mathrm{rank}(\rho)D$ state copies are necessary to fully reconstruct $\rho$  \cite{haah2017sample}. 
Constructive protocols \cite{wright2016tomography,haah2017sample} saturate this bound, but require entangled circuits and measurements that act on all state copies simultaneously. More tractable measurement procedures, where each copy of the state is measured independently, require an order of $\mathrm{rank}(\rho)^2 D$ measurements \cite{haah2017sample}. This more stringent bound is saturated by low rank matrix recovery \cite{flammia2012compressed,kueng2017lowrank,kueng2016low} and projected least squares estimation \cite{guta2018fast}. 

These results highlight an exponential bottleneck for tomography protocols that work in full generality: at least $D=2^n$ copies of an unknown $n$-qubit state are necessary. This exponential scaling extends to the computational cost associated with storing and processing the measurement data.

\paragraph{Matrix product state tomography} 

Restricting attention to highly structured subsets of quantum states sometimes allows for overcoming the exponential bottleneck that plagues general tomography. 
Matrix product state (MPS) tomography \cite{cramer2010efficient} is the most prominent example for such an approach. It only requires a polynomial number of samples, provided that the underlying quantum state is well approximated by a MPS with low bond dimension. In quantum many body physics this assumption is often justifiable \cite{lanyon2017efficient}. However, MPS representations of general states have exponentially large bond dimension. In this case, MPS tomography offers no advantage over general tomography.

\paragraph{Neural network tomography}

Recently, machine learning has also been applied to the problem of predicting features of a quantum systems. 
These approaches construct a classical representation of the quantum system by means of a deep neural network that is trained by feeding in quantum measurement outputs. 
Compared to MPS tomography, neural network tomography may be more broadly applicable \cite{gao2017efficient, torlai2018neural, carrasquilla2019reconstructing}. 
However, the actual class of systems that could be efficiently represented, reconstructed and manipulated is still not well understood.

\paragraph{Compressed classical description of quantum states}

To circumvent the exponential scaling in representing quantum states, Gosset et al. \cite{gosset2018compressed} have proposed a stabilizer sketching approach that compresses classical description of quantum states to a subexponential-sized description.
The stabilizer sketches bear some similarity with classical shadows.
However, stabilizer sketching requires a fully-characterized exponential-sized classical description of the quantum state, so it still suffer from an exponential scaling in the resources used in practice.

\paragraph{Direct fidelity estimation}

Direct fidelity estimation is a procedure that allows for predicting a single pure target fidelity $\langle \psi| \rho | \psi \rangle$ up to accuracy $\epsilon$. 
The best-known technique is based on few Pauli measurements that are selected randomly using importance sampling \cite{flammia2011direct}. 
The required number of samples depends on the target: it can range from a dimension-independent order of $1/\epsilon^2$
(if $| \psi \rangle$ is a stabilizer state) to roughly $2^n/\epsilon^4$ in the worst case.

\paragraph{Shadow tomography}

Shadow tomography aims at simultaneously estimating the probability associated with $M$ 2-outcome measurements up to accuaracy $\epsilon$: $p_i (\rho) = \mathrm{tr}(E_i \rho)$, where each $E_i$ is a positive semidefinite matrix whose with operator norm is at most one \cite{aaronson2018shadow,brandao2019sdp,aaronson2019gentle}. 
This may be viewed as a generalization of fidelity estimation.
The best existing result due to Aaronson et al. \cite{aaronson2019gentle} showed that
\begin{equation}
N= \tilde{\mathcal{O}} \left(\log (M)^2 \log (D)^2/\epsilon^8 \right)
\label{eq:aaronson_scaling}
\end{equation} copies of the unknown state\footnote{The scaling symbol $\tilde{\mathcal{O}}$ suppresses logarithmic expressions in other problem-specific parameters.}  suffice to achieve this task.
In a nutshell, his protocol is based on gently measuring the 2-outcome measurements one-by-one and subsequently (partially) reverting the perturbative effects a measurement exerts on quantum states. This task is achieved by explicit quantum circuits of exponential size that act on all copies of the unknown state simultaneously. 
This rather intricate procedure bypasses the no-go result advertised in Theorem~\ref{thm:main2}
and results in a sampling rate that is independent of the measurement in question -- only their cardinality $M$ matters.

\subsection{Summary of classical shadows}

It is fruitful to think of classical shadows as an alternative protocol for shadow tomography that emphasizes tractability. It also allows prediction of arbitrary linear target functions \eqref{eq:linear_predictions}.
Data acquisition is achieved by sequentially performing simple Clifford measurements on independent copies of the unknown state. This setup is motivated by the
abundant appearance of Clifford circuits in many practical applications, such
as fault-tolerant quantum computing \cite{nielsen2000quantum} and randomized benchmarking \cite{emerson2005randomized}.
The resulting data -- the \emph{classical shadow} of $\rho$ -- is comprised of stabilizer states that can be stored very efficiently. Prediction is done via
simple median of means estimation (Algorithm~\ref{alg:median_means}).
We equip this procedure with a rigorous statistical convergence guarantee:
\begin{equation}
N = \tilde{\mathcal{O}} \left( \frac{\log (M)}{\epsilon^2} \max_{1 \leq i \leq M} \mathrm{tr}(O_i^2) \right) \label{eq:our_scaling}
\end{equation}
Clifford measurements suffice to ensure $\epsilon$-accurate prediction of all target functions simultaneously. In contrast to Eq.~\eqref{eq:aaronson_scaling}, this scaling depends on the predictors. If each $O_i$ has (approximately) constant Hilbert-Schmidt norm, Eq.~\eqref{eq:our_scaling} improves over Aaronson's performance guarantees. This includes direct fidelity estimation ($O_1=| \psi \rangle \! \langle \psi|$) as a particularly relevant special case. Classical shadows extend the best case guarantee from \cite{flammia2011direct} to \emph{all} states -- an exponential improvement.
However, other target functions -- like Pauli observables -- have Hilbert-Schmidt norms that scale exponentially in the number of qubits. If this is the case, shadow tomography \eqref{eq:aaronson_scaling} does provide an exponential advantage.
This is not a flaw of classical shadows, but a fundamental feature of prediction schemes that are based on measuring state copies independently (Theorem~\ref{thm:main2}): tractable measurement procedures come at a price. 

We emphasize that tractability also extends to the prediction stage, especially if the target functions are fidelities with stabilizer states. This allowed us to support our theoretical findings with numerical experiments on quantum systems comprised of more than 160 qubits. 
Our findings do not only support our theoretical results.
We also observe that simple median of means estimation using classical shadows shows consistent improvement over existing neural network approach based on local SIC/tetrahedral measurements at predicting relevant, non-local target functions over different system sizes.

\subsection{Outlook}

\paragraph{Practical implementation}

We have designed classical shadows as a tool to tackle current challenges in quantum information processing. We have prioritized tractability in both quantum (Clifford circuits and computational basis measurements) and classical (stabilizer formalism) aspects of the protocol. This should in principle allow for a concrete implementation in state of the art quantum architectures. We are curious to test how this protocol performs in practice.

\paragraph{Extension to quantum channel tomography}

It would be highly interesting to extend classical shadow prediction to quantum channel tomography. Is it possible to create compressed classical sketches of unknown quantum channels that allow for accurately predicting their action on an exponential number of input states?

\paragraph{Synthesis with shadow tomography}

Relying on independent quantum measurements imposes fundamental restrictions on any prediction procedure. This restrictions can be overcome if one stores independent copies of the state in a quantum memory and performs joint quantum manipulations and measurements \cite{aaronson2018shadow,aaronson2019gentle}. 
Existing protocols that exploit this feature do not utilize classical representations of the state to make predictions. Instead, they jointly measure the targets directly. Is it possible to combine both ideas to create improved classical shadows with stronger performance guarantees?

\paragraph{Synthesis with machine learning}

Our numerical simulations highlight that prediction using classical shadows already offer advantages over existing methods without using any modern machine learning techniques.
An interesting direction would be to combine machine learning and classical shadows, e.g., by training neural networks using classical shadows (obtained from random Clifford measurements) as the input data.
While efficiently training neural networks with unitary-rotated measurements may be challenging, we believe the synthesis of machine learning and classical shadows would be fruitful in achieving improved performance with strong theoretical guarantees.

\subsection*{Acknowledgments:} The authors want to thank Victor Albert, Fernando Brand\~{a}o, John Preskill, Ingo Roth, Joel Tropp, and Thomas Vidick for valuable inputs and inspiring discussions. Leandro Aolita and Giuseppe Carleo provided helpful advice regarding presentation.
Our gratitude extends, in particular, to Joseph Iverson who helped us devising a numerical sampling strategy for toric code ground states.
HH is supported by the Kortschak Scholars Program.
RK acknowledges funding provided by the Office of Naval Research (Award N00014-17-1-2146) and the Army Research Office (Award W911NF121054).

\subsection*{Author Contributions:} All authors contributed equally in the work presented in this paper.

\subsection*{Competing interests:} The authors declare no competing interests.

\begin{appendix}

\section{Proof of Theorem~\ref{thm:main}}
\label{sec:thm1proof}

\subsection{The stabilizer formalism}

Clifford circuits were introduced by Gottesman \cite{gottesman1997stabilizer} and form an indispensable tool in  quantum information processing. Applications range from quantum error correction \cite{nielsen2000quantum}, to measurement-based quantum computation \cite{raussendorf2001measurementbased,briegel2009measurementbased} and randomized benchmarking \cite{emerson2005randomized,knill2008randomized,magesan2011randomized}.

For systems comprised of $n$ qubits, the Clifford group is generated by CNOT, Hadamard and phase gates. This results in a finite group of cardinality $2^{\mathcal{O}(n^2)}$ that maps (tensor products of) Pauli matrices to Pauli matrices upon conjugation. This underlying structure allows for efficiently storing and simulating Clifford circuits on classical computers -- a result commonly known as Gottesman-Knill theorem. 

If one applies a Clifford circuit $C$ to a computational basis state $|b \rangle$, one produces a \emph{stabilizer state} $|s \rangle = C |b \rangle$. 
A wide variety of conceptually and practically interesting quantum states belong into this category, e.g.\ GHZ states \cite{greenberger1989GHZ,hein2006entanglement}, cluster states \cite{nielsen2006cluster}  and toric code ground states \cite{kitaev2002toric}. Nonetheless, the Gottesman-Knill theorem asserts that keeping track of stabilizer states only requires limited classical resources, see Table~\ref{tab:stabilizer_cost} for a selection of important features.
Classical shadows correspond to lists of stabilizer states and enjoy all these benefits. 

\begin{table}
\begin{tabular}{ccc}
\textbf{task} & \textbf{cost} & \textbf{reference} \\
\hline
Sample a random Clifford circuit & $\mathcal{O}(n^2)$ classical bits of randomness   & \cite{koenig2014efficiently} \\
 &  $\mathcal{O}(n^3)$ classical pre-processing time & \\
Implement Clifford circuit & $\mathcal{O}(n^2 /\log (n))$ elementary 2-qubit gates & \cite{aaronson2004improved} \\
Store stabilizer state & $\mathcal{O}(n^2)$ classical bits &  \cite{aaronson2004improved} \\
Overlaps between stabilizer states & $\mathcal{O}(n^2)$ classical flops
 & \cite{garcia2012efficient,kueng2015qubit}
\end{tabular}
\caption{Cost associated with various Clifford/stabilizer operations on $n$ qubits: all relevant subroutines for obtaining Clifford shadows \eqref{eq:classical_shadow} scale favorably in the problem dimension $D=2^n$.}
\label{tab:stabilizer_cost}
\end{table}

For the purpose of Theorem~\ref{thm:main} it is instructive to re-interpret the measurement procedure that gives rise to classical shadows. 
The union of all possible stabilizer states comprises a single, tomographic complete quantum measurement (POVM) with cardinality $| \mathrm{STAB}|=2^n \prod_{j=1}^n \left(2^j+1 \right) =2^{\mathcal{O}(n^2)}$ \cite[Corollary~21]{gross2006hudson}:
\begin{equation}
\mathrm{STAB} = \left\{ \tfrac{2^n}{| \mathrm{STAB}|} |s \rangle \! \langle s|: \textrm{$s$ is stabilizer state} \right\}. \label{eq:stab_measurement}
\end{equation}

\begin{fact} \label{fact:classical_shadow}
Performing a randomly selected Clifford circuit and subsequently measuring the computational basis is equivalent to performing the stabilizer measurement \eqref{eq:stab_measurement}.
\end{fact}

The set of all $n$-qubit stabilizer states comprises a complex projective 3-design \cite{webb2015clifford,zhu2017clifford,kueng2015qubit}: sampling uniformly from it reproduces the first 3 moments of the uniform (Haar) measure over all possible pure states:
\begin{equation*}
\tfrac{1}{| \mathrm{STAB}|}\sum_{|s\rangle \textrm{ stabilizer state}} \left(|s \rangle \! \langle s| \right)^{\otimes k} = \int_{|v \rangle \textrm{ unif}}  \left(|v \rangle \! \langle v| \right)^{\otimes k} \mathrm{d} \mu (v) \quad \textrm{for} \quad k=1,2,3.
\end{equation*}
The right hand side of this equation can be evaluated explicitly by using techniques from representation theory, see e.g.\ \cite[Sec.~3.5]{gross2015partial}.
This in turn yields closed-form expressions for averages of quadratic and cubic polynomials over stabilizer states. For instance, let $A,B,C$ be Hermitian $2^n \times 2^n$ matrices and assume that $C$ is traceless. Then,
\begin{align}
\tfrac{1}{| \mathrm{STAB}|} \sum_{|s \rangle\textrm{ stabilizer state}} \langle s | A |s \rangle \langle s | B |s \rangle =&
\int_{\textrm{unif}}\langle u| A| u \rangle  \langle u| B |u \rangle \mathrm{d} \mu (u) =
\frac{\mathrm{tr}(A) \mathrm{tr}(B) + \mathrm{tr}(AB)}{(2^n+1)2^n}, \label{eq:2design} \\
\tfrac{1}{| \mathrm{STAB}|} \sum_{|s \rangle\textrm{ stabilizer state}}
\langle s| A |s \rangle \langle s |C |s \rangle^2 =&
\int_{\mathrm{unif}}\langle u|A|u \rangle \langle u|C|u \rangle^2 \mathrm{d} \mu (u)
=
\frac{\mathrm{tr}(A) \mathrm{tr}(C^2) + 2 \mathrm{tr}(A C^2)}{(2^n+2) (2^n+1) 2^n}.
\label{eq:3design}
\end{align}
This feature has immediate consequences for linear feature estimation with classical shadows. 

\begin{lemma} \label{lem:unbiased}
Suppose that we perform the stabilizer measurement \eqref{eq:stab_measurement} on an unknown $n$-qubit state $\rho$. Upon receiving the outcome $|s \rangle$, we set $\hat{\rho} = (2^n+1)|s \rangle \! \langle s|-\mathbb{I}$. 
Then, for any Hermitian $2^n \times 2^n$ matrix $O$, $\hat{o}(1,1) = \mathrm{tr} \left( O \hat{\rho} \right)$ is a random variable that obeys
\begin{equation}
\mathbb{E}\left[\hat{o}(1,1)\right] = \mathrm{tr}(O \rho) \quad \textrm{and} \quad  \mathrm{Var} \left[ \hat{o}(1,1)\right] \leq 3 \mathrm{tr} (O^2).
\label{eq:mean+variance}
\end{equation}
\end{lemma}

\begin{proof}
 Born's rule asserts that the probability for obtaining the outcome $|s \rangle$ is $\tfrac{2^n}{| \mathrm{STAB}|} \langle s| \rho |s \rangle$. Combine this with Eq.~\eqref{eq:2design} and $\mathrm{tr}(\rho)=1$ to establish the first claim:
\begin{align*}
\mathbb{E}\left[ \hat{o}(1,1) \right]+ \mathrm{tr}(O) =& (2^n+1) \mathbb{E} \left[ \langle s| O |s \rangle \right]= \tfrac{(2^n+1)2^n}{| \mathrm{STAB}|} \sum_{s} \langle s| \rho |s \rangle \langle s| O |s \rangle = \mathrm{tr}(O\rho) + \mathrm{tr}(\rho) \mathrm{tr}(O).
\end{align*}
For the second claim,
 note that the variance of any random variable $X$ is invariant under shifts: $\mathrm{Var}[X] = \mathrm{Var}[X-\beta]$ for all $\beta \in \mathbb{R}$. Let $O_0 = O - \tfrac{\mathrm{tr}(O)}{2^n} \mathbb{I}$ be the traceless part of $O$ and set 
 $\hat{o}_0(1,1) = \mathrm{tr} \left( \tilde{O} \hat{\rho}\right)=o(1,1) - \mathrm{tr}(O)/2^n$.
Shift-invariance then implies $\mathrm{Var}\left[\hat{o}(1,1)\right]=\mathrm{Var}\left[ \hat{o}_0 (1,1) \right] \leq \mathbb{E} \left[ \mathrm{tr} \left( O_0 \hat{\rho}\right)^2 \right]$ and the variance bound readily follows from applying Eq.~\eqref{eq:3design} and utilizing  $\mathrm{tr} \left( \rho O^2_0\right) \leq \mathrm{tr}(O_0^2 ) \leq \mathrm{tr}(O^2)$. 
\end{proof}

\subsection{Median of means estimation}

Lemma~\ref{lem:unbiased} sets the stage for successful feature estimation via classical shadows. A single stabilizer sample, i.e.\ a classical shadow of size $N=1$, correctly predicts any linear feature in expectation. 
Convergence to this desired expectation value can be boosted by 
forming empirical averages of multiple independent repetitions.
The \emph{empirical mean} is the canonical example for such a procedure: perform $N$ independent stabilizer measurements and upon receiving outcomes $|s_1 \rangle, \ldots, |s_N \rangle$ set $\hat{\rho}=\tfrac{1}{N} \sum_{i=1}^N |s_i \rangle \! \langle s_i| - \mathbb{I}$. For context, we point out that this is just the linear inversion estimator associated with the stabilizer measurement \eqref{eq:stab_measurement} \cite{guta2018fast}. Subsequently, construct
\begin{equation}
\hat{o}(N,1) = \mathrm{tr} \left( O \hat{\rho} \right) = \tfrac{1}{N} \sum_{i=1}^N \left( (2^n+1) \langle s_i | O |s_i \rangle - \mathrm{tr}(O) \right), \label{eq:mean_estimator}
\end{equation}
in order to predict the value of a linear feature $\mathrm{tr}(O\rho)$. 
According to Lemma~\ref{lem:unbiased}, each summand is an independent random variable with bounded variance and correct expectation value. Convergence to the expectation value can be controlled by applying Chebyshev's inequality: 
$\mathrm{Pr} \left[ \left| \hat{o}(N,1) - \mathrm{tr}(O\rho) \right| > \epsilon \right] \leq \epsilon^{-2} \mathrm{Var}(\hat{o}(N,1)) \leq 3 \mathrm{tr}(O^2) / (N \epsilon^2)$.
In order to achieve a deviation probability of (at most) $\delta$, the number of samples must scale like $N = 3 \mathrm{tr}(O^2)/(\delta \epsilon^2)$. 
While the scaling in variance and accuracy is optimal, the dependence on $1/\delta$ is particularly bad and reflects outlier corruption. Individual contributions to \eqref{eq:mean_estimator} can assume exponentially large values. And, although rare, such contributions can completely distort the sample mean estimator.

\emph{Median of means} \cite{nemirovski1983medianmeans,jerrum1986medianmeans} 
is a conceptually simple trick that addresses this issue.
Instead of using all $N$ samples to construct a single empirical mean \eqref{eq:mean_estimator}, construct $K$ independent independent linear inversion estimators: $\hat{\rho}_{(k)} = \frac{2^n+1}{\lfloor N/K \rfloor} \sum_{i = (k-1) \lfloor N/K \rfloor+1}^{k\lfloor N/K \rfloor} U_i^\dagger | \hat{b}_i \rangle \! \langle \hat{b}_i | U_i - \mathbb{I}$.
Subsequently, estimate a linear feature $o (\rho) = \mathrm{tr}(O \rho)$ by setting
\begin{equation}
\hat{o}(N,K) = \mathrm{median} \left\{ \mathrm{tr} \left( O \hat{\rho}_{(1)} \right),\ldots,\mathrm{tr} \left( O \hat{\rho}_{(K)} \right) \right\}. \label{eq:median_estimator}
\end{equation}
This estimator is much more robust towards outlier corruption. Indeed, $| \hat{o}(N,K)-\mathrm{tr}(O \rho) | > \epsilon$ if an only if more than half of the empirical means individually deviate by more than $\epsilon$. The probability associated with such an undesirable event decreases exponentially with the number of batches $K$. The Chernoff inequality  \cite{chernoff1952inequality} captures this feature and establishes an exponential improvement over mean estimation in terms of failure probability.

\begin{theorem}[Median of means estimation] \label{thm:median_theorem}
Let $X$ be a random variable with mean $\mu$ and variance $\sigma^2$ and choose a number of batches $K$. Then, $N = 34K \sigma^2 /\epsilon^2$ independent copies of $X$ suffice to construct a median of means estimator $\hat{\mu}(N,K)$ that obeys
\begin{equation*}
\mathrm{Pr} \left[ \left| \hat{\mu}(N,K) - \mu \right| > \epsilon \right] \leq 2 \mathrm{e}^{-K/2} \quad \textrm{for all} \quad \epsilon >0.
\end{equation*}
\end{theorem}

\subsection{Feature prediction using classical shadows and median of means}
\label{sec:thmrestate}

According to Fact~\ref{fact:classical_shadow}, classical shadows are synonymous with independent repetitions of the stabilizer measurement \eqref{eq:stab_measurement}.
We choose to rephrase our first main result in this language to maintain coherence with the previous two subsections.

\begin{theorem}[Detailed restatement of Theorem~\ref{thm:main}]
\label{thm:mainrestate}
Fix a collection $O_1,\ldots,O_M$ of $2^n \times 2^n$ Hermitian matrices and set $\epsilon, \delta \in [0,1]$. Then, performing independent stabilizer measurements \eqref{eq:stab_measurement} on
$
N = 204 \log (2M/\delta)\max_i \mathrm{tr}(O_i^2)/\epsilon^2
$
 copies of an unknown $n$-qubit state $\rho$ suffice to accurately predict all function evaluations $\mathrm{tr}(O_i \rho)$ via median of means estimation \eqref{eq:median_estimator}: with probability at least $1-\delta$,
\begin{equation*}
\left| \hat{o}_i (N,\log (M/\delta)) - \mathrm{tr}(O_i \rho ) \right| \leq \epsilon \quad \textrm{for all} \quad 1 \leq i \leq M.
\end{equation*}
\end{theorem}
\begin{proof}
According to Lemma~\ref{lem:unbiased}, each $\hat{o}_i (1,1) = (2^n+1) \langle s| O_i |s \rangle - \mathrm{tr}(O)$ is a random variable with mean $
\mu=\mathrm{tr}(O_i \rho)$ and variance $\sigma^2 \leq 3  \mathrm{tr}(O_i^2)$. 
In turn, Theorem~\ref{thm:median_theorem} asserts that the associated median of means estimator \eqref{eq:median_estimator} with $K= 2 \log (2M/ \delta)$ batches and a total of $N= 34 \times 2 \log (2M/ \delta) \times  3 \max_i \mathrm{tr}(O_i^2))/\epsilon^2$ samples obeys
\begin{equation*}
\mathrm{Pr} \left[ \left| \hat{o}_i (N,K) - \mathrm{tr}(O_i \rho ) \right| > \epsilon \right] \leq 2 \mathrm{e}^{-\log (2M/\delta)} = \delta/M.
\end{equation*}
Apply a union bound over all $M$ function predictions to complete the proof.
\end{proof}

Finally, we consider an improvement of Theorem~\ref{thm:mainrestate} that applies to the case of very many samples, i.e.\ $M >2^{2^n}$.  In this regime, it is better to perform feature prediction using a modified protocol. 

\begin{theorem}[Improvement of Theorem~\ref{thm:mainrestate}]
Instantiate notation from Theorem~\ref{thm:mainrestate} and set $B= \max_i \mathrm{tr}(O_i^2)$.
There exists a concrete protocol that uses the measurement result of independent stabilizer measurements \eqref{eq:stab_measurement} on
\begin{equation}
N = \mathcal{O}\Bigg(\min\Bigg(\frac{B\log(M)}{\epsilon^2}, \frac{B^2 D}{\epsilon^4} \log\Bigg(\frac{\sqrt{B}}{\epsilon}\Bigg), \frac{BD^2}{\epsilon^2} \log\Bigg(\frac{\sqrt{BD}}{\epsilon}\Bigg) \Bigg)\Bigg)
\label{eq:improved_size}
\end{equation}
copies of an unknown $n$-qubit state $\rho$ to accurately predict all function evaluations $\mathrm{tr}(O_i \rho)$ within $\epsilon$ error.
\end{theorem}

These improvements match with the rigorous lower bounds derived in Theorem~\ref{thm:main2restate} up to $\log$-factors.

\begin{proof}
The three terms in Eq.~\eqref{eq:improved_size} correspond to three variants of a quantum feature prediction protocol. The first term corresponds to the original protocol in Theorem~\ref{thm:mainrestate}.
The protocols corresponding to the second and third terms arise from covering all pure states in $\mathbb{C}^D$ with a net $\mathcal{C}_\eta$ of fineness $\eta$. I.e.\ for any state $| \psi \rangle$, there exists a net point $|u \rangle \in \mathcal{C}_\eta$ that obeys $\| |u \rangle \! \langle u| - | \psi \rangle \! \langle \psi| \|_1 \leq \eta$.
The fineness of this net distinguishes the two cases: we choose $\eta = \epsilon^2 /8 B$ for the second protocol and $\eta = \epsilon/(2 \sqrt{BD})$ for the third.
A volumetric counting argument asserts that the cardinality of the net obeys $|\mathcal{C}_\eta| \leq (1+2/\eta)^{2D}$.

Subsequently, we use feature prediction via classical shadows and median of mean to estimate target fidelities with the net points $\langle u| \rho |u \rangle$, not the actual targets. Theorem~\ref{thm:mainrestate} asserts that this requires $N= \mathcal{O}(\log (|\mathcal{C}_\eta|)/\eta^2)$ many stabilizer measurements. Inserting the volumetric bounds for both cases reproduces the second and third term in Eq.~\eqref{eq:improved_size}.
Subsequently, we approximate the actual linear features with linear combinations of net points. Let $\hat{u}_u$ denote the estimate for $\langle u| \rho |u \rangle$, where $u \in \mathcal{C}_\eta$.

For estimation protocol II, we fix $O$, apply an eigenvalue decomposition (PCA)  $O=\sum_{i=1}^D w_i | \omega_i \rangle \! \langle \omega_i|$ and set $O_\epsilon = \sum_{i:|w_i|>\epsilon/2} w_i | \omega_i \rangle \! \langle \omega_i|$. By construction, this truncation obeys $\left| \mathrm{tr}(O_\epsilon \rho) - \mathrm{tr}(O \rho) \right| \leq \epsilon /2$ for any underlying state. Next, we approximate each eigenvector $|\omega_i \rangle$ by the closest net-state $|u_i \rangle \in \mathcal{C}_\eta$ and set $\hat{o} = \sum_{i: |w_i|>\epsilon/2} w_i \hat{o}_{u_i}$. The particular choice of $\eta$ then ensures
\begin{equation*}
    \left| \hat{o} = \mathrm{tr}(O \rho) \right|\leq \left| \hat{o} - \mathrm{tr}(O_\epsilon \rho) \right|+ \frac{\epsilon}{2} \leq \sum_{i:|w_i| >\epsilon/2}|w_i| \frac{\epsilon^2}{4B}+\frac{\epsilon}{2} \leq \epsilon.
\end{equation*}
The last inequality combines Cauchy-Schwarz with the fact that there are at most $4B/\epsilon^2$ different eigenvalues that obey $|w_i| > \epsilon/2$.

In estimation protocol III, we simply omit the truncation to sufficiently large eigenvalues and set $\hat{o}=\sum_{i=1}^D w_i \hat{o}_{u_i}$. We can then use $\left| \langle \omega_i| \rho | \omega_i \rangle - \hat{o}_{u_i} \right| \leq \epsilon /\sqrt{BD}$ to conclude
\begin{equation*}
\left| \hat{o} - \mathrm{tr}(O \rho) \right| \leq \sum_i | w_i| \frac{\epsilon}{\sqrt{BD}} \leq \epsilon.
\end{equation*}
All three protocols are always valid. So, we can always choose the one with the smallest size $N$, and perform the corresponding protocol. This justifies the minimum in Eq.~\eqref{eq:improved_size}.
\end{proof}

\section{Proof of Theorem~\ref{thm:main2}}
\label{sec:proofthm2}

Before stating the content of the statement, we need to introduce some additional notation.
In quantum mechanics, the most general notion of a quantum measurement is a POVM (positive operator-valued measure). A $D$-dimensional POVM $F$ consists of a collection $F_1,\ldots,F_N$ of positive semidefinite matrices that sum up to the identity matrix: $\langle x|F_i|x \rangle \geq 0$ for all $x \in \mathbb{C}^D$ and $\sum_i F_i = \mathbb{I}$. 
The index $i$ is associated with different potential measurement outcomes and Born's rule asserts $\mathrm{Pr} \left[ i| \rho \right]=\mathrm{tr}(F_i \rho)$ for all $1 \leq i \leq M$ and any $D$-dimensional quantum state $\rho$.

\subsection{Detailed statement and proof idea}
\begin{theorem}[Detailed restatement of Theorem~\ref{thm:main2}]
\label{thm:main2restate}
Fix a sequence of POVMs $F^{(1)},\ldots,F^{(N)}$.
Suppose that given any $M$ features $0 \preceq O_1, O_2, \ldots, O_M \preceq I$ with $\max_i\Big(\norm{O_i}_2^2\Big) \leq B$, there exists a machine (with arbitrary runtime as long as it always terminates) that can use the measurement outcomes of $F^{(1)}, \ldots, F^{(N)}$ on $N$ copies of an unknown quantum state $\rho$ to $\epsilon$-accurately predict 
$$\Tr(O_1 \rho), \ldots, \Tr(O_M \rho)$$
with high probability.
Then, necessarily
$$N \geq \Omega \Bigg(\min\Bigg(\frac{B \log(M)}{\epsilon^2}, \frac{B^2 D}{\epsilon^4 }, \frac{B D^2}{\epsilon^2 } \Bigg) \Bigg).$$
\end{theorem}

We emphasize that this fundamental lower  bound exactly matches the constructive results in Theorem~\ref{thm:mainrestate}.
Feature prediction with classical shadows is sample-optimal among any prediction procedure based on agnostic measurements.

\paragraph*{Proof idea:} We adapt a versatile proof technique for establishing information-theoretic lower bounds on tomographic procedures that is originally due to Flammia et al.~\cite{flammia2012quantum}, see also  \cite{haah2017sample,roth2019recovering} for adaptations and refinements.
The key idea is to integrate an idealized prediction procedure into a quantum communication protocol. Alice encodes a classical bit-string into a quantum state and subsequently sends independent copies of this code-state to Bob. 
Bob then uses the device as a black-box to decode and retrieve the original bit-string. Intuitively, one expects a relation between message length $\log_2 (M)$ (the size of the communicated bit-string) and the number of code state copies $N$ that Bob receives: More state copies allow for a more accurate reconstruction facilitating the use of larger code-books comprised of more densely packed code-states. Techniques from information theory allow for making this intuition precise.
The result is a fundamental relation between the number of quantum state samples $N$ and $M$ -- the number of linear features that any black box device can predict up to a certain accuracy.

In order to obtain strong bounds, we equip this communication scenario with a twist that pinpoints an essential feature of linear feature prediction with classical shadows: the actual measurements are completely independent from the features one aims to predict. 
This twist is personified by a third party -- we call them Loki\footnote{In Norse mythology, Loki is infamous for mischief and trickery. However, not entirely malicious, he often comes around and tries to remedy dire consequences of his actions in the last minute.} -- who tampers with the quantum communication protocol. More precisely, he takes Alice's code states and randomly rotates them before presenting them to Bob (Loki's mischief). 
Bob is then forced to perform his quantum measurements on rotated quantum states. Only after the data acquisition phase is completed, Loki confides in Bob and provides him with a full classical description of the unitary he applied earlier (Loki's redemption). 
We refer to Figure~\ref{fig:commillus} for an illustration.

We emphasize that quantum feature prediction with classical shadows can cope with Loki's action (simply undo Loki's transformation post measurement by instructing the machine to predict a rotated set of linear features), while most existing approaches cannot. 
Therefore it should not come as a surprise that linear feature prediction with agnostic quantum measurements -- like classical shadows -- is constrained by more stringent information-theoretic bounds than methods that attune measurements to predictions  -- like shadow tomography \cite{aaronson2018shadow,aaronson2019gentle}.

\begin{figure}[t]
    \centering
    \includegraphics[width=1.0\textwidth]{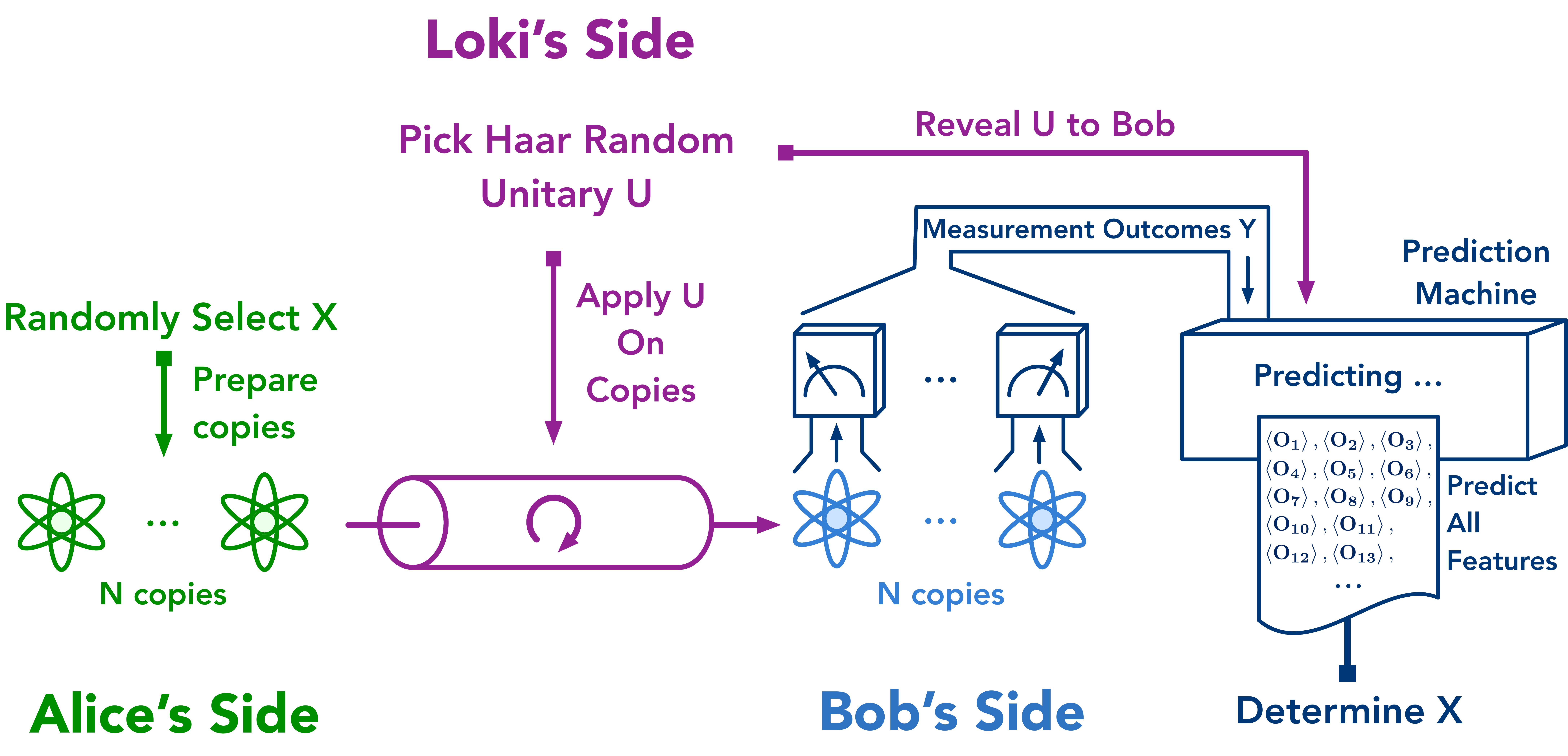}
    \caption{\emph{Illustration of the communication protocol behind Theorem~\ref{thm:main2restate}:}
    Two parties -- Alice and Bob -- devise a protocol that allows them to communicate classical bit strings: Alice encodes a bit string $X$ in a quantum state and sends $N$ independent copies of the state to Bob. Bob performs quantum measurements and uses a black box device (e.g.\ classical shadows) to decode Alice's original message. An unpredictable trickster -- Loki -- tampers with this procedure by randomly rotating Alice's quantum states en route to Bob. He only reveals his actions after Bob has completed the measurement stage of his protocol.}

    \label{fig:commillus}
\end{figure}

\subsection{Description of the communication protocol}

For now, assume $M \leq \exp(D / 32)$ and show how Alice can communicate any integer in $\{1, \ldots, M\}$ to Bob.
Alice and Bob first agree on a codebook for encoding any integer selected from $\{1, \ldots, M\}$ in a quantum state of dimension $D$.
The quantum states in the codebook are $\rho_1, \ldots, \rho_M$.
Alice and Bob also agree on a set of linear features $O_1, \ldots, O_M$ that satisfies
$$\Tr(O_i \rho_i) \geq \max_{j \neq i} \Tr(O_j \rho_i) + 3\epsilon.$$
Therefore $O_i$ can be used to identify individual code states.
The communication protocol between Alice and Bob is now apparent:
\begin{enumerate}
    \item Alice randomly selects an integer $X$ from $\{1, \ldots, M\}$.
    \item Alice prepares $N$ copies of the code-state $\rho_X$ according associated to $X$ and sends them to Bob.
    \item Bob performs POVMs $F^{(i)}$ on individual states and receives a string of measurement outcomes $Y$. 
    \item Bob inputs $Y$ into the feature prediction machine to estimate $\Tr(O_1 \rho_X), \ldots, \Tr(O_M \rho_X)$.
    \item Bob finds $\overline{X}$ that has the largest $\Tr(O_{\overline{X}} \rho_X)$.
\end{enumerate}

The working assumption is that the feature prediction machine can estimate $\Tr(O_1 \rho_X), \ldots, \Tr(O_M \rho_X)$ within $\epsilon$-error and high success probability.
This in turn ensures that this plain communication protocol is mostly successful, i.e.\ 
$\overline{X} = X$ with high probability.
In words:  Alice can transmit information to Bob, when no adversary is present.

We now show how they can still communicate safely in the presence of an adversary -- Loki -- that randomly rotates the transmitted code-states en route: $\rho_X \mapsto U \rho_X U^\dagger$ and $U$ is a Haar-random unitary.

This random rotation affects the measurement outcome statistics associated with the fixed POVMs $F^{(1)},\ldots,F^{(N)}$. Each element of $Y=\left[Y^{(1)},\ldots,Y^{(N)}\right]$ is now a random variable that depends on both $X$ and $U$.
After Bob has performed the quantum measurements to obtain $Y$, the adversary confesses to Bob and reveals the random unitary~$U$.
While Bob no longer has any copies of $\rho_X$, he can still incorporate precise knowledge of $U$ by instructing the machine to predict linear features $U O_1 U^\dagger, \ldots, U O_M U^\dagger$ instead of the original $O_1, \ldots, O_M$.
This reverses the effect of the original unitary transformation, because $\Tr(U O_i U^\dagger U \rho_X U^\dagger) = \Tr(O_i \rho_X)$.
This modification renders the original communication protocol stable with respect to Loki's actions. Alice can still send any integer in $\{1, \ldots, M\}$ to Bob with high probability.

\subsection{Information-theoretic analysis}
\label{sec:infotheoanalysis}

The following arguments are based on basic concepts from information theory. We refer to standard textbooks for details.

The communication protocol is guaranteed to work with high probability, ensuring that Bob's recovered message $\hat{X}$ equals Alice's input $X$ with high probability. Moreover, we assume that Alice selects her message uniformly at random. Fano's inequality then implies
$$I(X : \overline{X}) = H(X) - H(X | \overline{X}) \geq \Omega(\log(M)),$$
where $I(X : \overline{X})$ is the mutual information, and $H(X)$ is the Shannon.
By assumption, Loki chooses the unitary roatation $U$ uniformly at random, regardless of the message $X$. This implies $I(X:U)=0$ and, in turn
$$I(X : \overline{X}) \leq I(X : \overline{X}, U) = I(X : U) + I(X : \overline{X} | U) = I(X : \overline{X} | U).$$
For fixed $U$, $\overline{X}$ is the output of the machine that only takes into account the measurement outcomes $Y$. The data processing inequality then implies
$$I(X : Y | U) \geq I(X : \overline{X} | U) \geq I(X : \overline{X}) \geq \Omega(\log(M)).$$
Recall that $Y$ is the measurement outcome of the $N$ POVMs $F_1, \ldots, F_N$. We denote the measurement outcome of $F_k$ as $Y_k$.
Because $Y_1, \ldots, Y_N$ are random variables that depend on $X$ and $U$,
\begin{align*}
    I(X : Y | U) & = H(Y_1, \ldots, Y_N | U) - H(Y_1, \ldots, Y_N | X, U)\\
    & \leq H(Y_1 | U) + \ldots + H(Y_N | U) - H(Y_1, \ldots, Y_N | X, U)\\
    & = \sum_{k=1}^N \Big( H(Y_k | U) - H(Y_k | X, U) \Big) = \sum_{k=1}^N I(X: F_k \mbox{ on } U\rho_X U^\dagger | U).
\end{align*}
The second to last equality uses the fact that when $X, U$ are fixed, $Y_1, \ldots, Y_N$ are independent.
To obtain the best lower bound, we want $I(X: F_k \mbox{ on } U\rho_X U^\dagger | U)$ to be as small as possible.
More precisely, we demand
\begin{equation}
\label{eq:codestates}
I(X: F_k \mbox{ on } U\rho_X U^\dagger | U) \leq \frac{36 \epsilon^2}{B}, \forall k,
\end{equation}
and are going to justify this relation in the following section. 
Assuming that this relation holds, we have established a connection between $M$ and $N$: $\Omega(\log(M)) \leq I(X: Y | U) \leq 36 N \epsilon^2 / B$ and, therefore,
$N \geq \Omega\Big(B\log(M) / \epsilon^2\Big)$.
This establishes the claim in Theorem~\ref{thm:main2restate} for the case $M \leq \exp(D / 32)$.

\subsection{Detailed construction of quantum encoding and linear prediction decoding}
We now construct a codebook $\rho_1, \ldots, \rho_M$ and linear features $0 \preceq O_1, O_2, \ldots, O_M \preceq I$ with $\max_i\Big(\norm{O_i}_2^2\Big) \leq B$ that obey two key properties:
\begin{enumerate}
    \item the code states $\rho_1,\ldots,\rho_M$ obey the technical requirement displayed in Eq.~\eqref{eq:codestates} holds.
    \item the linear features $O_1,\ldots,O_M$ are capable of identifying concrete code states:
    \begin{equation}
    \label{eq:decoding}
    \Tr(O_i \rho_i) \geq \max_{j \neq i} \Tr(O_j \rho_i) + 3\epsilon \quad \textrm{for all} \quad 1 \leq i \leq M.
    \end{equation}
\end{enumerate}
The second condition requires each $\rho_i$ to be distinguishable from $\rho_1, \ldots, \rho_M$ via linear features $O_i$.
The first condition, on the contrary, requires $\rho_X$ to convey as little information about $X$ as possible.
The general idea would then be to create distinguishable quantum states that are, at the same time, very similar to each other.

In order to achieve these two goals,
we choose $M$ (recall that we focus on the case $M \leq \exp(D / 32)$) rank-$B/4$ subspace projectors that obey $\Tr(P_i P_j) / r < 1 / 2$ for all $i \neq j$.
The probabilistic method asserts that such a projector configuration exists, see
Lemma~\ref{lem:probabexist} below. Now, we set
$$\rho_i = (1 - 3\epsilon) \frac{I}{D} + 3 \epsilon \frac{4 P_i}{B}, \quad \textrm{and} \quad  O_i = 2 P_i, \quad \textrm{for all} \quad 1 \leq i \leq M.
$$
It is easy to check that this construction meets the requirement displayed in Eq.~\eqref{eq:decoding}. 
The other condition -- Eq.~\eqref{eq:codestates} is verified in 
Lemma~\ref{lem:randbound} below.
Along with the information-theoretic analysis provided earlier, this concludes the proof of the lower bound for the case where $M \leq \exp(D / 32)$ (case 1).

We now consider the case where $B / 9\epsilon^2 \leq D$ (case 2) and $B / 9\epsilon^2 > D$ (case 3). 
For these cases, we modify the codebook size:
\begin{align*}
\textrm{case 2:}\quad  & M' = \exp\Bigg(\min\Bigg(\log(M), \frac{B D}{1152\epsilon^2}\Bigg)\Bigg),\\ 
\textrm{case 3:} \quad & M' = \exp\Bigg(\min\Bigg(\log(M), \frac{ D^2}{128}\Bigg)\Bigg).
\end{align*}
These adjustments have a pragmatic reason: since $M$ is very large, we may we may not be able to find enough desirable codestates using the probabilistic method.
From Lemma~\ref{lem:probabexist}, we can only find $M' \leq M$ rank-$r$ subspace projectors such that $\Tr(P_i P_j) / r < 1/2, \forall i \neq j$, where $r = B / 36 \epsilon^2$ for case 2 and $r = D / 4$ for case 3.
The quantum states in the codebook and the corresponding linear features are
\begin{align*}
  & \rho_i = \frac{36\epsilon^2}{B} P_i, && \hat{O}_i = 6\epsilon P_i, & \mbox{ for case 2,}\\
  & \rho_i = \Bigg(1 - 3\epsilon \sqrt{\frac{D}{B}}\Bigg) \frac{I}{D} + 3\epsilon \sqrt{\frac{D}{B}} \frac{P_i}{D / 4}, && \hat{O}_i = 2\sqrt{\frac{B}{D}} P_i, & \mbox{ for case 3.}
\end{align*}
Using Lemma~\ref{lem:probabexist} and \ref{lem:randbound}, we can see that the two desired properties of the codebook are both satisfied.
Following the same information-theoretic analysis using a different codebook, we have
$N \geq \Omega\Big(\min\Big(B\log(M) / \epsilon^2, B^2D /\epsilon^4 \Big)\Big)$ when $B / 9\epsilon^2 \leq D$ and $N \geq \Omega\Big(\min\Big(B\log(M) / \epsilon^2, BD^2 / \epsilon^2 \Big)\Big)$ when $B / 9\epsilon^2 > D$.
Together, we have
$$N \geq \Omega \Bigg(\min\Bigg(\frac{B \log(M)}{\epsilon^2}, \frac{B^2 D}{\epsilon^4 }, \frac{B D^2}{\epsilon^2 } \Bigg) \Bigg).$$
Actually, case 2 and 3 alone already yield the desired result, but we walk through the proof focusing on case 1 because it puts an emphasis on $M \leq \exp(D / 32) = \exp(2^n / 32)$, which is the most relevant case.

\begin{lemma}
\label{lem:probabexist}
If $M \leq \exp(rD / 32)$ and $D \geq 4r$, then there exists $M$ rank-$r$ subspace projectors $P_1, \ldots, P_M$ such that
$$\Tr(P_i P_j) / r < 1 / 2, \forall i \neq j.$$
\end{lemma}
\begin{proof}
We find the subspace projectors using probabilistic argument. We randomly choose $M$ rank-$r$ subspaces according to the unitarily invariant measure in the Hilbert space, the Grassmannian, and bound the probability that the randomly chosen subspaces do not satisfy the condition.
For a pair of fixed $i\neq j$, we have
$$\Pr\Bigg[\frac{1}{r} \Tr(P_i P_j) \geq \frac{1}{2}\Bigg] \leq \exp\Bigg(- r^2 f\Bigg(\frac{D}{2r} - 1\Bigg)\Bigg) < \exp\Bigg(- \frac{rD}{16}\Bigg),$$
where we make use of  \cite[Lemma~6]{haah2017sample} in the first inequality and $f(z) = z - \log(1+z) > z/4$ for all $z \geq 1$ in the second inequality.
A union bound then asserts
$$\Pr\bigg[\exists i \neq j, \frac{1}{r} \Tr(P_i P_j) \geq \frac{1}{2}\bigg] < M^2 \exp\Bigg(- \frac{rD}{16} \Bigg) \leq 1.$$
Because the probability is less than one, there must exist $P_1, \ldots, P_{M'}$ that satisfy the desired property.
\end{proof}

\begin{lemma}
\label{lem:randbound}
For any POVM measurement $F$, let $\vec{F}$ be the vector of POVM elements in $F$. And consider a set of quantum states $\{\rho_1, \ldots, \rho_M\}$ such that $\rho_i = (1-\alpha) \frac{I}{D} + \alpha \frac{P_i}{r}$, where $P_i$ is a rank-$r$ subspace projector. Consider $U$ sampled from Haar measure, and $X$ sampled from $\{1, \ldots, M\}$ uniformly at random, then the accessible information of $X$ from the POVM measurement $F$ on $U\rho_X U^\dagger$ conditioned on $U$ satisfies
$$I(X: F \mbox{ on } U\rho_X U^\dagger | U) \leq \frac{\alpha^2}{r}.$$
In case 1 $(M \leq \exp(D / 32))$, $\alpha = 3 \epsilon, r = B / 4$.
In case 2 $(B/9\epsilon^2 \leq D)$, $\alpha = 1, r = \frac{B}{36\epsilon^2}$. In case 3 $(B/9\epsilon^2 > D)$, $\alpha = 3\epsilon \sqrt{D / B}, r = D / 4$.
\end{lemma}
\begin{proof}
First of all, let us decompose all POVM elements $\{F_1, \ldots, F_{l}\}$ to rank-$1$ elements $F' = \big\{w_i D \ket{v_i} \bra{v_i}\big\}_{i=1}^{l'}$, where $l \leq l'$. We can perform measurement $F$ by performing measurement with $F'$: when we measure a rank-$1$ element, we return the original POVM element the rank-$1$ element belongs to.
Using data processing inequality, we have $I(X: F \mbox{ on } U\rho_X U^\dagger | U) \leq I(X: \tilde{F} \mbox{ on } U\rho_X U^\dagger | U)$.
From now on, we can consider the POVM $\vec{F}$ to be $\big\{w_i D \ket{v_i} \bra{v_i}\big\}_{i=1}^l$.
We have 
$$\Tr\Big(\sum_i w_i D \ket{v_i} \bra{v_i}\Big) = \Tr(I) = D \implies \sum_i w_i = 1.$$
Let us define the probability vector $\vec{p} = \Tr(U\rho_1 U^\dagger \vec{F}),$ so $p_i = w_i D \bra{v_i} U\rho_1 U^\dagger \ket{v_i}.$
And the expression we hope to bound satisfies $I(X: F \mbox{ on } U\rho_X U^\dagger | U) = I(X, U : F \mbox{ on } U\rho_X U^\dagger) - I(U : F \mbox{ on } U\rho_X U^\dagger) \leq I(X, U : F \mbox{ on } U\rho_X U^\dagger)$ using the chain rule and the nonnegativity of mutual information.
We now bound
\begin{align*}
    I(X, U : F \mbox{ on } U\rho_X U^\dagger) = & H\Big( \sum_{X=1}^M \frac{1}{M} E_U [ \Tr(U \rho_X U^\dagger \vec{F}) ] \Big) - \sum_{X=1}^M \frac{1}{M} E_U \Big[ H\Big(\Tr(U \rho_X U^\dagger \vec{F}) \Big)\Big] \\
	= & H\Big( \Tr(E_U [U \rho_1 U^\dagger] \vec{F}) \Big) - E_U \Big[H\Big(\Tr(U \rho_1 U^\dagger \vec{F}) \Big)\Big] \\
	= & \sum_i - (\E_U p_i) \log(\E_U p_i) + \E_U [p_i \log p_i] \\
	\leq & \sum_i - (\E_U p_i) \log(\E_U p_i) + \E_U\Big[p_i \log(\E_U p_i) + p_i \frac{p_i - \E_U p_i}{\E_U p_i}\Big] \\
	= & \sum_i \frac{\E_U[p_i^2] - \E_U[p_i]^2}{\E_U[p_i]}.
\end{align*}
The second equality uses the fact that $E_U f(U \rho_X U^\dagger) = E_U f(U \rho_1 U^\dagger), \forall X$ which follows from the fact that $\forall X, \exists U_X, \rho_X = U_X \rho_1 U_X^\dagger$.
The inequality uses the fact that $\log(x)$ is concave, so $\log(x) \leq \log(y) + \frac{x-y}{y}$.
Using properties of Haar random unitary stated in Equation~\eqref{eq:2design} , we have
$$\E_U[p_i] = w_i, \,\, \E_U[p_i^2] = w_i^2 \frac{D}{(D+1)} \Bigg(1 + \frac{1}{D} + \alpha^2\Big(\frac{1}{r} - \frac{1}{D}\Big)\Bigg).$$
Therefore we have
$$\frac{\E_U[p_i^2] - \E_U[p_i]^2}{\E_U[p_i]} = w_i \alpha^2 \frac{D}{D+1} \Big(\frac{1}{r} - \frac{1}{D}\Big) \leq \frac{w_i \alpha^2}{r}.$$
Hence we have arrived at
$$I(X: F \mbox{ on } U\rho_X U^\dagger | U) \leq \sum_i \frac{\E_U[p_i^2] - \E_U[p_i]^2}{\E_U[p_i]} \leq \frac{\alpha^2}{r}.$$
\end{proof}

\section{Details regarding numerical experiments}
\label{sec:detail}

\paragraph*{Direct fidelity estimation and comparison with neural network tomography:}
To perform efficient classical simulation on larger system size (more than 100 qubits), we mainly focus on stabilizer states, because classical simulation of hundreds of qubits is possible using Gottesman-Knill theorem and the improved classical algorithm \cite{aaronson2004improved}.
We implement our protocol for constructing classical shadow of quantum states as follows.
\begin{enumerate}
    \item Sample a Clifford unitary $U$ from the Clifford group using the algorithm proposed in \cite{koenig2014efficiently}. This Clifford unitary is parameterized by $(\alpha, \beta, \gamma, \delta, r, s)$ which fully characterize its action on Pauli operators:
    \begin{equation*}
    UX_j U^\dagger = (-1)^{r_j} \Pi_{i=1}^n X_i^{\alpha_{ji}} Z_i^{\beta_{ji}} \quad \textrm{and} \quad 
    UZ_j U^\dagger = (-1)^{s_j} \Pi_{i=1}^n X_i^{\gamma_{ji}} Z_i^{\delta_{ji}} \quad \textrm{for all} \quad j.
    \end{equation*}
    \item
    Given a unitary $U$ parameterized by $(\alpha, \beta, \gamma, \delta, r, s)$, we can apply $U$ on any stabilizer state by changing the stabilizer generators and the destabilizers as defined in \cite{aaronson2004improved}.
    \item Measurement in the Z-basis can be done using the standard algorithm provided in \cite{aaronson2004improved}.
\end{enumerate}

The operational definition of mixed states allows us to extend this formalism. Mixed states  arise from sampling a pure state ensemble: $\rho=\sum_i p_i | \psi_i \rangle \! \langle \psi_i|$.

\begin{figure}[t]
    \centering
    \includegraphics[width=0.84\textwidth]{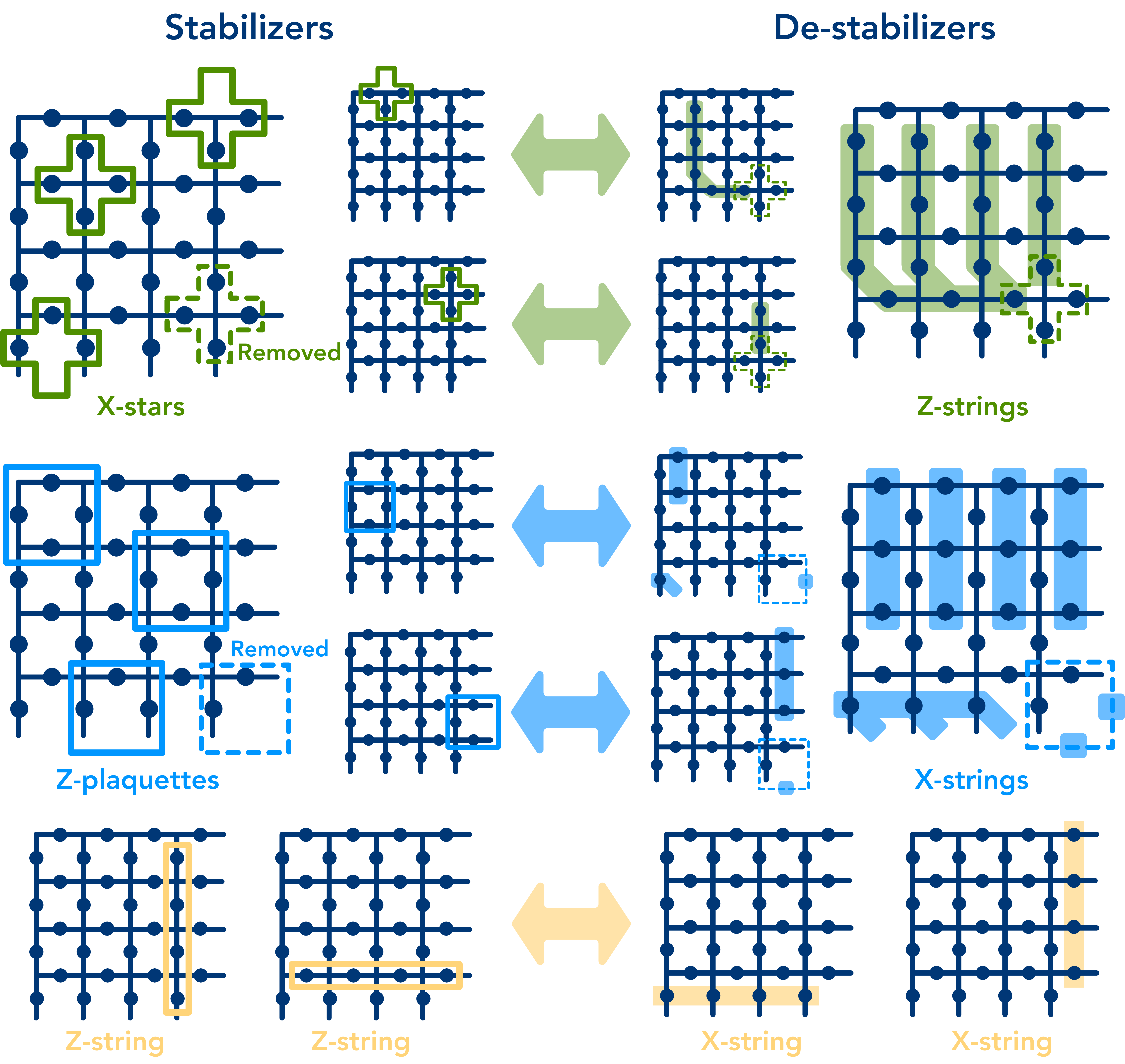}
    \caption{Stabilizers and de-stabilizers of toric code that encodes $\ket{00}$.}
    \label{fig:tableau}
\end{figure}

For neural network quantum state tomography, we use the open-source code provided by the authors \cite{carrasquilla2019reconstructing}.
The main part is to generate training data, i.e.\ simulating measurement outcomes.
For pure and noisy GHZ state, we use the tetrahedral POVM \cite{carrasquilla2019reconstructing}.
For toric code, we use the Psi2 POVM (which is a measurement in the Z-basis).
Note that measuring in the Z-basis is not a tomographically complete measurement, but we found machine learning models to perform better using Psi2. This is possibly because the pattern is much more obvious (closed-loop configurations) and the figure of merit used in NNQST is a classical fidelity.

A concrete algorithm for creating training data for pure GHZ states is included in the open-source implementation provided by the authors \cite{carrasquilla2019reconstructing}. It uses matrix product states to simulate quantum measurements efficiently.
The training data for noisy GHZ states is a slight modification of the existing code. With probability $1-p$, we sample a measurement outcome from the original
state $|\psi_{\mathrm{GHZ}}^+ \rangle = \tfrac{1}{\sqrt{2}} (|0 \rangle^{\otimes n} + |1 \rangle^{\otimes n})$.
And with probability $p$, we sample 
a measurement outcome from 
 $| \psi_{\mathrm{GHZ}}^-\rangle=\frac{1}{\sqrt{2}} (\ket{0}^{\otimes n} - \ket{1}^{\otimes n})$ (phase error).
Since the figure of merit is the fidelity with the pure GHZ state in both pure and noisy GHZ experiment, we reuse the implementation provided in \cite{carrasquilla2019reconstructing}.

Creating training data for toric code is somewhat more involved.
The goal is to sample a closed-loop configuration on a 2D torus uniformly at random.
This can again be done using classical simulations of stabilizer states \cite{aaronson2004improved}.
The main technical detail is to create a tableau that contains both the stabilizer and the de-stabilizer for the state in question. The rich structure of the toric code renders this task rather easy.
The stabilizers are the X-stars and the Z-plaquettes, with two Z-strings over the two loops of the torus.
The de-stabilizer of each stabilizer is a Pauli-string that anticommutes with the stabilizer, but commutes with other stabilizers and other de-stabilizers.
The full set of stabilizers and de-stabilizers for the toric code can be seen in Figure~\ref{fig:tableau}.

\vspace{1.5em}

\paragraph*{Potential difficulty in learning some quantum states:}
In the experiment, we have seen that neural network quantum state tomography based on deep generative models seems to have difficulty learning the toric code. We provide further analysis and construct a simple class of quantum states where efficient learning of the quantum state from the measurement data would violate a well-known computational hardness conjecture.
First of all, each measurement of the toric code in the Z-basis would be a random bit-string, where most bits are sampled uniformly at random from $\{0, 1\}$ and other bits are binary functions of the randomly sampled bits.
We now consider a simple class of quantum states that shares a similar property. Given $a \in \{0, 1\}^{n-1}$ and $f_a(x) = \sum_i a_i x_i$ (mod $2$), we define $\ket{a} = \frac{1}{\sqrt{2^{n-1}}}\sum_{x \in \{0, 1\}^{n-1}} \ket{x} \ket{f_a(x)}$, which can be created by preparing $\ket{+}$ on the first $n-1$ qubits and $\ket{0}$ on the $n$-th qubit followed by CNOT gates between $i$-th qubit and $n$-th qubit for every $a_i = 1$.
Measuring $\ket{a}$ in the Z-basis is the same as sampling the first $n-1$ bits $x$ uniformly at random, and the last bit is given deterministically as $f_a(x)$.
We then consider a noisy version of $\ket{a}$ and define
$$\rho_a = (1 - \eta) \ket{a}\bra{a} + \eta I / 2^n,$$
for some $0 < \eta < 1$.

One of the most widely used conjectures for building post-quantum cryptography is the hardness of learning with error (LWE) \cite{regev2009lattices}.
LWE considers the task of learning a linear $n$-ary function $f$ over a finite ring from noisy data samples $(x, f(x)+\eta)$ where $x$ is sampled uniformly at random and $\eta$ is some independent error.
An efficient learning algorithm for LWE will be able to break many post-quantum cryptographic protocals that are believed to be hard even for quantum computers.
The simplest example of LWE is called learning parity with error, where $f(x) = \sum_i a_i x_i$ (mod $2$) for $x \in \{0, 1\}^n$ and some unknown $a \in \{0, 1\}^n$. Learning parity with error is also conjectured to be computationally hard \cite{blum2003noise}.
Since learning $a$ from Z-basis measurements on $\rho_a$ is the same as the task of learning parity with error, it is unlikely there will be a neural network approach that can learn $\rho_a$ efficiently.

\vspace{1.5em}

\paragraph*{Witnesses for tri-partite entanglement: }
For this experiment, we consider 3-qubit states of the form 
\begin{equation*}
\ket{\psi} = U_A \otimes U_B \otimes U_C  | \psi_{\mathrm{GHZ}}^+ \rangle,
\end{equation*}
where $U_A, U_B, U_C$ are random single-qubit rotations.
We consider the simplest form of entanglement witness:
\begin{equation*}
\hat{O} = \alpha I - V_A \otimes V_B \otimes V_C | \psi_{\mathrm{GHZ}}^+ \rangle \! \langle \psi_{\mathrm{GHZ}}^+| V_A^\dagger \otimes V_B^\dagger \otimes V_C^\dagger.
\end{equation*}
The scalar parameter $\alpha$ and the single-qubit unitaries $ V_A, V_B, V_C$ parameterize the different candidates.
For witnessing genuine tripartite entanglement, we only need $\alpha$ to be $0.5$. To witness GHZ-type entanglement, we need $\alpha$ to be $0.75$.
Recall that we want $\bra{\psi}\hat{O}\ket{\psi} < 0$ to be a witness.
This means there will be more $V_A, V_B, V_C$ that can witness genuine tripartite entanglement, and less that can witness GHZ-type entanglement.

We generate random $V_A, V_B, V_C$ as candidates.
For direct measurements, we consider the number of total experiments required to estimate every $\hat{O}$ up to $0.1$-error.
Note that the number of required samples may vary from witness to witness. It depends on the variance associated with the estimation.

For parameter estimation with classical shadows, we first determine the total number $N$ of measurement repetitions required to $0.1$-accurately predict \emph{all} entanglement witnesses simultaneously, see Theorem~\ref{thm:main}. 
Subsequently, we test the performance by estimating all target witnesses.

Because the system size is small (three qubits), we simulate quantum experiments classically by storing and processing all $2^3=8$ amplitudes.

\end{appendix}

\bibliography{ref}
\bibliographystyle{abbrv}

\end{document}